\documentclass[a4paper,fleqn,12pt]{cas-dc}

\usepackage[numbers,sort&compress]{natbib}
\usepackage{amsmath,amssymb}
\usepackage{booktabs}
\usepackage{tikz}
\usetikzlibrary{arrows.meta,positioning,calc,fit,shapes.geometric}
\usepackage{graphicx}
\usepackage{xcolor}
\usepackage{algorithm}
\usepackage{algpseudocode}

\graphicspath{{./images/}}

\newcommand{\NonlinearValidationN}{500}

\newtheorem{theorem}{Theorem}
\newtheorem{lemma}[theorem]{Lemma}
\newtheorem{proposition}[theorem]{Proposition}
\newtheorem{corollary}[theorem]{Corollary}
\newdefinition{assumption}{Assumption}
\newdefinition{remark}{Remark}
\newproof{proof}{Proof}

\makeatletter
\def\@opargbegindefinition#1#2#3{%
  \trivlist
  \item[]\noindent
  {\bfseries #1\ #2\ (#3).}\ \upshape
}
\def\@opargbegintheorem#1#2#3{%
  \trivlist
  \let\baselinestretch\@blstr
  \item[]\noindent
  {\bfseries #1\ #2\ (#3).}\ \itshape
}
\makeatother

\begin{document}
\hypersetup{pdftitle={Uniform High-Probability ISS Tubes for Sampled-Data State Estimation},pdfauthor={Jerzy Baranowski},pdfkeywords={state estimation, sampled-data systems, input-to-state stability, high-probability bounds, observer design, nonlinear systems},hypertexnames=false}
\let\WriteBookmarks\relax
\def\floatpagepagefraction{1}
\def\textpagefraction{.001}
\setlength{\mathindent}{0pt}
\setlength{\emergencystretch}{1.5em}

\shorttitle{High-probability ISS tubes for sampled-data estimation}
\shortauthors{Baranowski}

\title[mode=title]{Uniform High-Probability ISS Tubes for Sampled-Data State Estimation}
\tnotemark[1]
\tnotetext[1]{Dedicated to Professor Wojciech Mitkowski on the occasion of his 80th birthday.}

\author[1]{Jerzy Baranowski}[orcid=0000-0003-3313-581X]
\cormark[1]
\ead{jb@agh.edu.pl}
\credit{Conceptualization, Methodology, Formal analysis, Software, Validation, Investigation, Visualization, Writing -- original draft, Writing -- review and editing}

\affiliation[1]{organization={Department of Automatic Control and Robotics, AGH University of Krakow},
            addressline={al. Mickiewicza 30},
            city={Krakow},
            postcode={30-059},
            country={Poland}}
\cortext[1]{Corresponding author}

\begin{abstract}
State estimates used in sampled monitoring and automation need bounds that remain valid between measurements. We develop a finite-horizon input-to-state-stability tube and observer co-design framework for continuous-time observers driven by sampled and held outputs. The sampled-data error model separates process disturbance, sampled measurement noise, and intersample mismatch. A horizon-level disturbance-envelope event is transferred through an ISS estimate to simultaneous containment of the complete error trajectory. Quadratic dissipation inequalities yield ellipsoidal and componentwise tubes, and semidefinite co-design minimizes normalized tube width across the three channels. A structured nonlinear extension preserves known nonlinear channels. Co-design reduces the worst normalized half-width by $31\%$ in a linear compartment benchmark and by a factor of $22.4$ in a flexible-joint benchmark.
\end{abstract}


\maketitle
\hypersetup{pdfsubject={arXiv manuscript with appended appendix material},pdfauthor={Jerzy Baranowski}}

\section{Introduction}

State estimation is a central component of monitoring and automation systems in which decisions are made from incomplete, noisy, and delayed information. In process plants, mechatronic devices, networked sensors, and embedded diagnostic modules, the physical state evolves continuously while measurements are acquired at discrete instants. The estimator must therefore propagate information between sensor updates. For alarms, constraint checking, supervision, and diagnostic escalation, the relevant uncertainty often concerns the whole continuous-time evolution over a finite interval, including the intersample parts of the trajectory.

This paper uses the term \emph{tube} for such an interval-wise uncertainty object. Formally, a tube is a time-indexed family of sets around the observer trajectory, for example
\begin{equation}
\begin{aligned}
    \mathcal T(t)=\{x:{}&(x-\hat x(t))^\top P(x-\hat x(t))\\
    &\leq \rho(t)\},\qquad t\in[0,T].
\end{aligned}
\end{equation}
or, in componentwise form, intervals $|x_i(t)-\hat x_i(t)|\leq r_i(t)$. A trajectory is contained in the tube when the corresponding inclusion holds for every $t\in[0,T]$. The probability statement studied here is therefore a simultaneous statement over the full horizon. This contrasts with marginal statements at individual time instants.

The sampled-data setting gives a concrete reason to distinguish these viewpoints. The observer correction uses the most recent sampled output, while the plant state continues to evolve. The estimation error is then driven by process disturbance, sampled measurement noise, and the mismatch between the current output and the held output. A finite-horizon tube for such an observer has to account for these channels together, and it has to state how a disturbance-side probability budget is transferred to an estimation-error bound.

This distinction is operational in measurement and automation systems. A pointwise covariance band may be appropriate for reporting local estimator uncertainty, but an alarm, supervisory constraint, or diagnostic escalation rule often asks whether the unmeasured state could have left an admissible region at any time since the last update. The proposed construction is therefore intended as a robustness-oriented monitoring bound used alongside statistical estimator uncertainty, not as a replacement for Kalman or moving-horizon filtering.

\subsection{Related work and positioning}

Kalman filtering and moving-horizon estimation (MHE) provide widely used stochastic state-estimation tools. Kalman's original filtering paper established the linear Gaussian prediction and filtering framework that underlies the Kalman-filter comparator used later in this paper \cite{kalman1960new}. Haseltine and Rawlings compared extended Kalman filtering and MHE in nonlinear process-estimation settings, which is close to the estimator-comparison role used in the numerical section \cite{haseltine2005critical}. Rao, Rawlings, and Mayne developed stability theory for constrained nonlinear MHE, which motivates the use of constrained MHE as a serious comparator rather than a purely heuristic baseline \cite{rao2003constrained}. Liu formulated MHE for nonlinear systems with bounded uncertainties \cite{liu2013robust}. Knuefer and M\"uller studied robust global asymptotic stability for nonlinear full-information and moving-horizon estimation \cite{knuefer2023mhe}. Schiller et al. constructed a Lyapunov function for robust MHE stability analysis \cite{schiller2023lyapunov}. Alessandri considered robust MHE when the optimization problem is solved imperfectly \cite{alessandri2025imperfect}. Recent outlier-oriented MHE work treats multiple measurement outliers in linear discrete-time systems \cite{liu2024multipleoutliers} and measurement outliers in nonlinear systems \cite{liu2025outliers}. Robust Bayesian inference has also been proposed for MHE, providing a distributional counterpart to deterministic robustness arguments \cite{cao2025robustbayes}. These methods provide the estimator-side comparators used later in the numerical study.

Guaranteed state estimation under bounded uncertainty is developed in set-membership, ellipsoidal, interval, and positive-system approaches. Milanese and Vicino surveyed optimal estimation for dynamic systems with set-membership uncertainty \cite{milanese1991setmembership}. Kurzhanski developed guaranteed state estimation using ellipsoidal techniques \cite{kurzhanski1994guaranteed}, and Chernousko provided a systematic treatment of ellipsoidal state estimation for dynamical systems \cite{chernousko2005ellipsoidal}. Loukkas et al. designed set-membership observers using ellipsoidal invariant sets, which is close in spirit to tube-based guaranteed estimation \cite{loukkas2017setmembership}. Farina and Rinaldi give the standard theory and applications of positive linear systems \cite{farina2000positive}, while Kaczorek gives a complementary systems-theoretic account of positive one- and two-dimensional systems \cite{kaczorek2002positive}. Zhang and Shen survey interval-observer design from the positive-system viewpoint \cite{zhang2023intervalsurvey}. Dinh and Tran give a recent systematic interval-observer construction for linear systems \cite{dinh2025interval}, and Thabet et al. address interval observers for continuous linear systems with unknown inputs and bounded disturbances \cite{thabet2026interval}. These lines of work are relevant because the linear benchmark has a positive physical state space and because the proposed tubes are also guaranteed state-containment objects.

Input-to-state stability (ISS) provides the robustness language used in the present paper to pass from input bounds to estimation-error bounds. Sontag introduced the ISS property as a way to express stability under persistent inputs \cite{sontag1995input}. Sontag and Wang provided equivalent characterizations of ISS \cite{sontag1996characterizations}, and Sontag later summarized the basic concepts and results in a form widely used in nonlinear control \cite{sontag2008iss}. Observer analysis for sampled or limited measurements has developed in parallel. Karafyllis and Kravaris constructed global exponential observers for two nonlinear system classes \cite{karafyllis2012global}. Ahmed-Ali et al. developed global exponential sampled-data observers for nonlinear systems with delayed measurements \cite{ahmedali2013sampled}. Mazenc, Andrieu, and Malisoff considered continuous-discrete observer design for time-varying nonlinear systems \cite{mazenc2015continuous}. Recent work includes sampled-data observers for linear Kuramoto--Sivashinsky systems with non-local output \cite{karafyllis2024kuramoto}, aperiodic-measurement observer designs \cite{huff2024aperiodic}, sample-based nonlinear detectability for discrete-time systems \cite{krauss2025detectability}, and sample-based detectability with continuous-time MHE \cite{krauss2026continuous}. These literatures motivate the assumptions used in Sections~\ref{sec:iss-tubes}--\ref{sec:codesign}; the new step is to combine them into a finite-horizon tube construction for sampled-data observer design.

\subsection{Contribution and organization}

The main contribution is a sampled-data tube construction and observer co-design method for finite-horizon monitoring. A preliminary conference paper introduced the disturbance-envelope interpretation and a linear compartment illustration \cite{baranowski2026pcc}. The present article extends that work in five directions. First, it states the ISS lifting as a horizon-level bridge from a disturbance envelope to simultaneous containment of the complete continuous-time error trajectory. Second, it separates process disturbance, sampled measurement noise, and held-output mismatch in the sampled-data observer error equation. Third, it derives ellipsoidal, Euclidean, and componentwise tubes from quadratic dissipation inequalities. Fourth, it formulates linear and structured nonlinear gain-metric co-design problems that minimize normalized componentwise tube widths. Fifth, it evaluates the construction on a nominally positive compartment model and a flexible-joint benchmark, with continuous-discrete Kalman filtering, extended Kalman filtering, Gaussian MHE, and robust Huber MHE used as estimator-side comparators where appropriate. The Kalman and MHE bands answer a pointwise estimator-uncertainty question, whereas the proposed tube answers a simultaneous containment question over the full finite horizon.

Figure~\ref{fig:theory-map} gives the map of the theoretical construction. The upper part of the diagram shows the abstract route from a disturbance envelope and an ISS estimate to a finite-horizon tube. The lower part shows the sampled-data specialization and the two co-design paths used later in the numerical study.

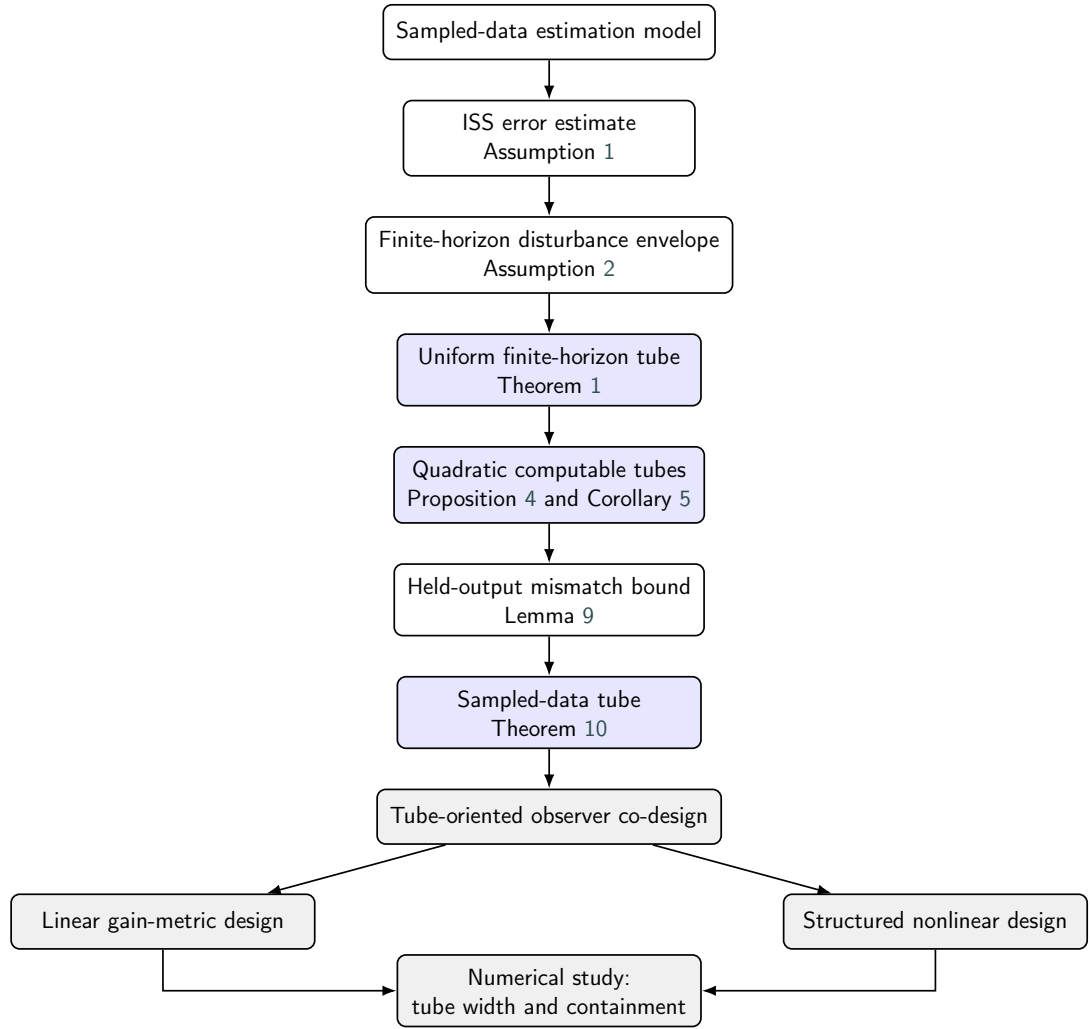
\begin{figure*}[pos=t]
\centering
\resizebox{0.82\textwidth}{!}{%
\begin{tikzpicture}[
    >=Latex,
    node distance=6.5mm,
    box/.style={draw, rounded corners, thick, align=center, minimum width=48mm, minimum height=9mm, inner sep=2.1mm},
    result/.style={draw, rounded corners, thick, align=center, minimum width=50mm, minimum height=9mm, inner sep=2.1mm, fill=blue!10},
    design/.style={draw, rounded corners, thick, align=center, minimum width=50mm, minimum height=9mm, inner sep=2.1mm, fill=gray!12},
    arr/.style={-Latex, thick}
]

\node[box] (model) {Sampled-data estimation model};
\node[box, below=of model] (iss) {ISS error estimate\\Assumption~\ref{ass:iss}};
\node[box, below=of iss] (env) {Finite-horizon disturbance envelope\\Assumption~\ref{ass:uniform-envelope}};
\node[result, below=of env] (tube) {Uniform finite-horizon tube\\Theorem~\ref{thm:uniform-hp-iss-tube}};
\node[result, below=of tube] (quad) {Quadratic computable tubes\\Proposition~\ref{prop:lyap-to-iss} and Corollary~\ref{cor:uniform-quadratic-tube}};
\node[box, below=of quad] (stale) {Held-output mismatch bound\\Lemma~\ref{lem:output-staleness}};
\node[result, below=of stale] (sampled) {Sampled-data tube\\Theorem~\ref{thm:sampled-data-uniform-tube}};
\node[design, below=of sampled] (design) {Tube-oriented observer co-design};
\node[design, below left=8mm and 10mm of design] (linear) {Linear gain-metric design};
\node[design, below right=8mm and 10mm of design] (nonlinear) {Structured nonlinear design};
\node[design, below=18mm of design] (study) {Numerical study:\\tube width and containment};

\draw[arr] (model) -- (iss);
\draw[arr] (iss) -- (env);
\draw[arr] (env) -- (tube);
\draw[arr] (tube) -- (quad);
\draw[arr] (quad) -- (stale);
\draw[arr] (stale) -- (sampled);
\draw[arr] (sampled) -- (design);
\draw[arr] (design) -- (linear);
\draw[arr] (design) -- (nonlinear);
\draw[arr] (linear.south) |- (study.west);
\draw[arr] (nonlinear.south) |- (study.east);

\end{tikzpicture}%
}
\caption{Map of the theoretical construction. The main vertical chain shows how a sampled-data estimation problem is converted into a finite-horizon error tube. The final branch shows how the sampled-data tube becomes an observer-design objective in the linear and structured nonlinear cases, which are then evaluated numerically.}
\label{fig:theory-map}
\end{figure*}

The paper is organized as follows. Section~\ref{sec:framework} introduces the sampled-data estimation setting and the notation. Section~\ref{sec:iss-tubes} develops the finite-horizon ISS tube. Section~\ref{sec:observer-conditions} gives observer contraction conditions used later. Section~\ref{sec:codesign} specializes the tube to sampled measurements and formulates the observer co-design problems. Section~\ref{sec:numerical} states the numerical objective and presents the two benchmark studies. Section~\ref{sec:discussion} discusses interpretation and limitations, and Section~\ref{sec:conclusions} concludes.

\subsection{Notation and abbreviations}
\label{subsec:notation}

The main symbols are summarized in Table~\ref{tab:notation}. For a symmetric matrix $P$, the relations $P\succ0$ and $P\succeq0$ denote positive definiteness and positive semidefiniteness, respectively, and $\mathbb S_{++}^n$ is the set of real symmetric positive definite $n\times n$ matrices. The Euclidean vector norm and induced spectral matrix norm are denoted by $\|\cdot\|_2$; the subscript is omitted when no ambiguity arises. For a measurable signal $s:[0,T]\to\mathbb R^q$,
\begin{equation}
\|s\|_{\infty,[0,T]}:=\operatorname*{ess\,sup}_{\tau\in[0,T]}\|s(\tau)\|_2.
\end{equation}
The classes $\mathcal K$, $\mathcal K_\infty$, and $\mathcal{KL}$ have their standard ISS meaning. The map $\sigma(t)$ denotes the most recent sampling instant. For a square matrix $M$,
\begin{equation}
\operatorname{He}(M):=M+M^\top.
\end{equation}
Thus, $\operatorname{He}(M)$ is the unscaled symmetric part; some references use $\operatorname{sym}(M)=\operatorname{He}(M)/2$.

\begin{table*}[pos=t]
\caption{Notation and abbreviations used in the sampled-data tube construction.}
\label{tab:notation}
\centering
\begin{tabular}{@{}p{0.22\textwidth}p{0.70\textwidth}@{}}
\toprule
Symbol or abbreviation & Meaning \\
\midrule
ISS & input-to-state stability \\
KF, EKF & Kalman filter and extended Kalman filter \\
MHE & moving-horizon estimation \\
LMI & linear matrix inequality \\
$x,\hat x,e$ & state, state estimate, and estimation error $x-\hat x$ \\
$w,v,\zeta_h$ & process disturbance, measurement noise, and held-output mismatch \\
$L,P$ & observer gain and Lyapunov metric \\
$T,\delta$ & finite horizon and admissible failure probability \\
$\bar d_T(\delta)$ & horizon-level disturbance envelope \\
$\rho(t)$ & Lyapunov-level tube radius \\
$s_i$ & physical scale used to normalize state component $i$ \\
\bottomrule
\end{tabular}
\end{table*}

\begin{figure*}[pos=htbp]
\centering

\resizebox{\textwidth}{!}{%
\begin{tikzpicture}[
    >=Latex,
    font=\small,
    block/.style={draw, rounded corners, thick, align=center, minimum height=10mm, minimum width=24mm},
    sblock/.style={draw, rounded corners, thick, align=center, minimum height=8mm, minimum width=19mm},
    tubeblock/.style={draw, rounded corners, thick, align=center, minimum height=8mm, minimum width=19mm, fill=blue!15},
    ann/.style={align=center},
    line/.style={-Latex, thick},
    dashedline/.style={-Latex, thick, dashed}
]

\node[block] (plant) {Plant\\$\dot{x}(t)=f(x,u,t)+w(t)$};
\node[sblock, right=18mm of plant] (sampler) {Sampler\\$y(t_k)$};
\node[sblock, right=15mm of sampler] (hold) {Sample-and-hold\\$y(\sigma(t))$};
\node[block, right=18mm of hold] (observer) {Observer / estimator\\$\dot{\hat{x}}(t)=\cdots$};

\node[sblock, below=14mm of observer] (error) {Error\\$e(t)=x(t)-\hat{x}(t)$};
\node[tubeblock, left=18mm of error] (tube) {Finite-horizon\\ISS tube};

\draw[line] ($(plant.west)+(-18mm,0)$) -- node[above] {$u(t)$} (plant.west);
\draw[line] (plant.east) -- node[above] {$y(t)$} (sampler.west);
\draw[line] (sampler.east) -- (hold.west);
\draw[line] (hold.east) -- node[above] {$y(\sigma(t))$} (observer.west);
\draw[line] (observer.east) -- ++(18mm,0) node[right,ann] {$\hat{x}(t)$};

\draw[line] (observer.south) -- (error.north);
\draw[line] (error.west) -- (tube.east);

\draw[dashedline] ($(plant.north)+(0,8mm)$) -- node[right] {$w(t)$} (plant.north);
\draw[dashedline] ($(sampler.north)+(0,8mm)$) -- node[right] {$v(t_k)$} (sampler.north);
\draw[dashedline] ($(hold.north)+(0,8mm)$) -- node[right,align=left] {staleness /\\held-output mismatch} (hold.north);

\node[
    draw,
    rounded corners,
    thick,
    align=left,
    anchor=north west,
    inner sep=3mm,
    text width=0.82\textwidth
] at ($(plant.south west)+(-2mm,-25mm)$)
{%
\textbf{(a)} Sampled-data estimation setting. The plant evolves continuously, while
its output is available only at sampling instants $t_k$ and is held between updates.
The observer uses the held measurement $y(\sigma(t))$ to generate $\hat{x}(t)$. The
resulting estimation error is influenced by process disturbance, measurement noise,
and intersample staleness.};

\end{tikzpicture}%
}

\vspace{4mm}

\resizebox{\textwidth}{!}{%
\begin{tikzpicture}[
    >=Latex,
    font=\small,
    line cap=round,
    line join=round
]
    \draw[->, thick] (0,0) -- (13.2,0) node[right] {$t$};
    \draw[->, thick] (0,0) -- (0,4.8);
    \node[anchor=east] at (-0.10,4.55) {signal value};

    \draw[densely dashed] (0.8,0) -- (0.8,4.4);
    \draw[densely dashed] (12.0,0) -- (12.0,4.4);
    \node[below] at (0.8,0) {$0$};
    \node[below] at (12.0,0) {$T$};
    \draw[<->, thick] (0.8,4.45) -- (12.0,4.45);
    \node[above] at (6.4,4.45) {finite horizon $[0,T]$};

    \foreach \x/\lab in {1.2/{t_0},2.8/{t_1},4.4/{t_2},6.0/{t_3},7.6/{t_4},9.2/{t_5},10.8/{t_6}} {
        \draw[densely dotted] (\x,0) -- (\x,3.8);
        \node[below] at (\x,0) {$\lab$};
    }

    \coordinate (L1) at (0.8,1.35);
    \coordinate (L2) at (2.0,1.55);
    \coordinate (L3) at (3.6,2.10);
    \coordinate (L4) at (5.2,2.45);
    \coordinate (L5) at (6.8,2.30);
    \coordinate (L6) at (8.4,2.00);
    \coordinate (L7) at (10.0,1.75);
    \coordinate (L8) at (12.0,1.55);

    \coordinate (U1) at (0.8,2.55);
    \coordinate (U2) at (2.0,2.75);
    \coordinate (U3) at (3.6,3.25);
    \coordinate (U4) at (5.2,3.55);
    \coordinate (U5) at (6.8,3.40);
    \coordinate (U6) at (8.4,3.05);
    \coordinate (U7) at (10.0,2.75);
    \coordinate (U8) at (12.0,2.50);

    \fill[blue!12]
        (L1) .. controls (1.2,1.45) and (1.6,1.45) .. (L2)
             .. controls (2.7,1.70) and (3.0,1.95) .. (L3)
             .. controls (4.2,2.30) and (4.7,2.40) .. (L4)
             .. controls (5.9,2.45) and (6.2,2.35) .. (L5)
             .. controls (7.5,2.20) and (7.8,2.05) .. (L6)
             .. controls (9.0,1.85) and (9.4,1.80) .. (L7)
             .. controls (10.8,1.65) and (11.2,1.58) .. (L8)
             --
        (U8) .. controls (11.0,2.58) and (10.8,2.65) .. (U7)
             .. controls (9.3,2.90) and (9.0,2.95) .. (U6)
             .. controls (7.8,3.22) and (7.4,3.30) .. (U5)
             .. controls (6.1,3.50) and (5.8,3.50) .. (U4)
             .. controls (4.8,3.45) and (4.4,3.35) .. (U3)
             .. controls (3.0,3.00) and (2.6,2.82) .. (U2)
             .. controls (1.6,2.68) and (1.2,2.60) .. (U1)
             -- cycle;

    \draw[very thick, blue]
        (1.0,1.95) .. controls (1.8,2.02) and (2.4,2.10) .. (3.2,2.45)
                  .. controls (4.0,2.75) and (4.8,2.95) .. (5.8,3.00)
                  .. controls (6.8,2.92) and (7.6,2.72) .. (8.4,2.50)
                  .. controls (9.2,2.28) and (10.2,2.18) .. (11.6,2.00);

    \draw[very thick, black]
        (1.0,1.70) .. controls (1.8,1.90) and (2.6,2.35) .. (3.3,2.70)
                  .. controls (4.0,3.00) and (4.8,3.15) .. (5.7,3.22)
                  .. controls (6.5,3.28) and (7.3,3.10) .. (8.0,2.88)
                  .. controls (9.0,2.48) and (10.0,2.18) .. (11.6,1.88);

    \foreach \x/\y in {1.2/1.76,2.8/2.42,4.4/3.05,6.0/3.18,7.6/2.88,9.2/2.36,10.8/2.00} {
        \filldraw[red] (\x,\y) circle (1.4pt);
    }

    \draw[red, thick, dashed]
        (1.2,1.76) -- (2.8,1.76) -- (2.8,2.42) -- (4.4,2.42) --
        (4.4,3.05) -- (6.0,3.05) -- (6.0,3.18) -- (7.6,3.18) --
        (7.6,2.88) -- (9.2,2.88) -- (9.2,2.36) -- (10.8,2.36) --
        (10.8,2.00) -- (12.0,2.00);

    \draw[-Latex, thick] (6.75,2.9) -- (6.75,3.21);
    \node[right] at (6.83,3.0) {$e(t)$};

    \begin{scope}[shift={(0.95,-0.75)}]
        \fill[blue!12] (0,0.08) rectangle (0.50,0.30);
        \draw[blue!45] (0,0.19) -- (0.50,0.19);
        \node[anchor=west] at (0.65,0.19) {ISS tube};

        \draw[very thick, black] (2.80,0.19) -- (3.35,0.19);
        \node[anchor=west] at (3.50,0.19) {true state $x(t)$};

        \draw[very thick, blue] (6.10,0.19) -- (6.65,0.19);
        \node[anchor=west] at (6.80,0.19) {estimate $\hat{x}(t)$};

        \filldraw[red] (9.85,0.19) circle (1.4pt);
        \node[anchor=west] at (10.05,0.19) {sample $y(t_k)$};
    \end{scope}
    \begin{scope}[shift={(4.65,-1.12)}]
        \draw[red, thick, dashed] (0,0.19) -- (0.55,0.19);
        \node[anchor=west] at (0.70,0.19) {held output $y(\sigma(t))$};
    \end{scope}

    \node[
        draw,
        rounded corners,
        thick,
        align=left,
        anchor=north west,
        inner sep=3mm,
        text width=0.82\textwidth
    ] at (-0.05,-1.45)
    {%
    \textbf{(b)} Schematic trajectory view. The true state $x(t)$ and the estimate
    $\hat{x}(t)$ evolve continuously, while measurements are available only at
    sample times $t_k$. The shaded band represents a deterministic or
    high-probability finite-horizon ISS tube around the estimate.};
\end{tikzpicture}%
}

\caption{Sampled-data estimation and finite-horizon ISS tubes. Panel (a) separates process disturbance, measurement noise, and held-output mismatch. Panel (b) illustrates continuous state and estimate trajectories, sampled measurements, and a simultaneous finite-horizon tube.}
\label{fig:conceptual-framework}
\end{figure*}

\section{Sampled-data estimation framework}
\label{sec:framework}

Consider a continuous-time plant
\begin{equation}
\dot x(t)=f(x(t),u(t),t)+w(t),
\qquad x(t)\in\mathbb R^n,
\label{eq:framework-plant}
\end{equation}
with an output map $h:\mathbb R^n\times\mathbb R_{\geq0}\to\mathbb R^m$. Measurements are available at a strictly increasing sequence $0=t_0<t_1<\cdots$,
\begin{equation}
y_k=h(x(t_k),t_k)+v_k,
\qquad h_k=t_{k+1}-t_k\leq \bar h.
\label{eq:framework-measurement}
\end{equation}
Define the most recent sampling instant
\begin{equation}
\sigma(t)=t_k,
\qquad t\in[t_k,t_{k+1}),
\end{equation}
and the zero-order-held measurement $\bar y(t)=y_k$. A continuous-time observer driven by this held signal has the generic form
\begin{equation}
\dot{\hat x}(t)=f(\hat x(t),u(t),t)
+L\bigl(\bar y(t)-h(\hat x(t),t)\bigr).
\label{eq:framework-observer}
\end{equation}
The estimation error is
\begin{equation}
e(t)=x(t)-\hat x(t).
\end{equation}

The held measurement is associated with $x(\sigma(t))$, whereas the plant has already evolved to $x(t)$. The resulting output mismatch is
\begin{equation}
\zeta_h(t)=h(x(t),t)-h(x(\sigma(t)),\sigma(t)).
\label{eq:framework-staleness}
\end{equation}
Substitution of $\bar y(t)=h(x(t),t)-\zeta_h(t)+\bar v(t)$ into \eqref{eq:framework-observer} gives
\begin{equation}
\begin{aligned}
\dot e(t)={}&f(x(t),u(t),t)-f(\hat x(t),u(t),t)\\
&-L\bigl(h(x(t),t)-h(\hat x(t),t)\bigr)\\
&+w(t)-L\bar v(t)+L\zeta_h(t).
\end{aligned}
\label{eq:framework-error}
\end{equation}
Thus the sampled-data implementation creates three distinct forcing channels: $w$, $-L\bar v$, and $L\zeta_h$. Their dependence on $L$ differs, which is central to the co-design problem developed later.

For the abstract analysis, these terms may be collected in a measurable input $d(t)$ and the error dynamics written as
\begin{equation}
\dot e(t)=f_e(e(t),d(t),t).
\label{eq:abstract-error-framework}
\end{equation}
Throughout the analysis, the known input $u$ is arbitrary but admissible and is suppressed from the notation. The statements are conditional on the usual Carath\'eodory well-posedness assumptions: for every locally essentially bounded $u$, every measurable locally essentially bounded $d$, and every initial condition under consideration, the error system has a unique absolutely continuous solution on the stated horizon. The solution map and the disturbance-envelope events are assumed measurable. The next section derives a finite-horizon implication for this abstract system and then specializes it to quadratic and sampled-data constructions.

The error representation in \eqref{eq:framework-error} is the point at which the sampled-data nature of the problem enters the theory. Section~\ref{sec:iss-tubes} first develops the finite-horizon tube for an abstract disturbance input. Section~\ref{sec:observer-conditions} then gives contraction conditions that provide the required deterministic ISS inequality. Section~\ref{sec:codesign} returns to \eqref{eq:framework-error}, treats the held-output mismatch as a separate design channel, and uses the resulting tube radius as an observer-design criterion.

\section{From ISS estimates to finite-horizon error tubes}
\label{sec:iss-tubes}

\subsection{Deterministic ISS bound}

This section starts from the ISS estimate that will be lifted to a finite-horizon tube. In ISS terminology, the aggregated input acting on the error determines an ultimate error bound together with a transient term. Sontag introduced this robustness property for nonlinear systems \cite{sontag1995input}. Sontag and Wang developed equivalent characterizations \cite{sontag1996characterizations}, and Sontag later summarized the role of ISS in nonlinear control \cite{sontag2008iss}. Here, the ISS estimate is used as the deterministic part of a probabilistic tube construction.

Let $e(t)\in\mathbb{R}^n$ denote the estimation error and assume its dynamics can be
written in the abstract form
\begin{equation}
\dot e(t) = f_e(e(t),d(t),t),
\label{eq:abstract-error}
\end{equation}
where $d(t)$ is an aggregated disturbance/noise input. For a measurable signal
$s:[0,t]\to\mathbb{R}^q$ we define the essential supremum norm
\[
\|s\|_{\infty,[0,t]} := \operatorname*{ess\,sup}_{\tau\in[0,t]}\|s(\tau)\|.
\]

\begin{assumption}[Deterministic ISS bound]
\label{ass:iss}
There exist functions $\beta\in\mathcal{KL}$ and $\gamma\in\mathcal{K}$ such that for
all $t\ge 0$,
\begin{equation}
\|e(t)\| \le \beta(\|e(0)\|,t) + \gamma(\|d\|_{\infty,[0,t]}).
\label{eq:iss-bound}
\end{equation}
\end{assumption}

Assumption~\ref{ass:iss} is a compact way of saying that the observer error is robustly stable with respect to the aggregated disturbance channel. Such inequalities are common in observer analysis. They follow immediately for exponentially stable linear error dynamics with additive inputs, arise from quadratic observer LMIs, and also appear in nonlinear observer designs based on high-gain and sampled-data constructions \cite{khalil2014highgain}, \cite{karafyllis2012global}, \cite{ahmedali2013sampled}. Sliding-mode and higher-order sliding-mode observer constructions provide another family of bounded-error or finite-time estimation mechanisms under bounded unknown inputs \cite{levant1998robust}, \cite{niederwieser2021higher}. In this paper the assumption is first used abstractly, and Sections~\ref{sec:observer-conditions} and~\ref{sec:codesign} then give concrete classes of observers and metrics that certify it.

\subsection{Uniform high-probability ISS tubes}

A fixed-time reading of an ISS estimate gives the error bound at one selected time. A tube over an entire interval requires a probability statement about the disturbance behavior throughout that interval. The next assumption expresses this requirement at the input side.

\begin{assumption}[Finite-horizon probabilistic disturbance envelope]
\label{ass:uniform-envelope}
Let $T>0$ and $\delta\in(0,1)$. There exists a deterministic quantity
$\bar d_T(\delta)\geq 0$ such that
\begin{equation}
    \mathbb{P}\left(
        \|d\|_{\infty,[0,T]}
        \leq \bar d_T(\delta)
    \right)
    \geq 1-\delta .
    \label{eq:uniform-disturbance-envelope}
\end{equation}
\end{assumption}

Assumption~\ref{ass:uniform-envelope} has a direct probabilistic reading: over the horizon $[0,T]$, the disturbance magnitude is bounded by the deterministic number $\bar d_T(\delta)$ with probability at least $1-\delta$. The value of $\bar d_T(\delta)$ may increase with $T$, because a longer horizon gives more opportunities for large disturbance or measurement-noise realizations. Once the event in Assumption~\ref{ass:uniform-envelope} occurs, the rest of the tube construction is deterministic.

Under Assumptions~\ref{ass:iss} and~\ref{ass:uniform-envelope}, the following result is the horizon-level lifting step used later in the sampled-data tube and co-design construction. It transfers one finite-horizon disturbance event to the complete estimation-error trajectory.

\begin{theorem}[Uniform ISS tube]
\label{thm:uniform-hp-iss-tube}
Suppose that Assumption~\ref{ass:iss} holds for every admissible disturbance realization
and that Assumption~\ref{ass:uniform-envelope} holds for a given
$T>0$ and $\delta\in(0,1)$. Assume that the initial estimation error $e(0)$
is deterministic. Define
\begin{equation}
\begin{aligned}
    r(\tau;T,\delta):={}&\beta\bigl(\|e(0)\|,\tau\bigr)\\
    &+\gamma\bigl(\bar d_T(\delta)\bigr),
    \quad \tau\in[0,T].
\end{aligned}
\label{eq:time-varying-uniform-radius}
\end{equation}
Then
\begin{equation}
\begin{aligned}
\mathbb{P}\bigl(&\|e(\tau)\|\leq r(\tau;T,\delta),\\
&\text{for all }\tau\in[0,T]\bigr)\geq1-\delta.
\end{aligned}
\label{eq:uniform-hp-tube}
\end{equation}

Consequently, with
\begin{equation}
    r_T(\delta)
    :=
    \sup_{\tau\in[0,T]} r(\tau;T,\delta),
    \label{eq:constant-uniform-radius}
\end{equation}
the estimation error satisfies
\begin{equation}
    \mathbb{P}\left(
        \sup_{\tau\in[0,T]}\|e(\tau)\|
        \leq r_T(\delta)
    \right)
    \geq 1-\delta .
    \label{eq:constant-uniform-tube}
\end{equation}
\end{theorem}

The theorem is the probabilistic lifting used throughout the paper. Once the horizon-level input event occurs, the ISS estimate gives containment of the complete error trajectory on the same horizon. Pointwise covariance or Laplace bands enter later as estimator-side comparators and answer a different uncertainty question.

\begin{proof}
Define the horizon-level event
\begin{equation}
    \mathcal{E}_{T,\delta}
    :=
    \left\{
        \omega:
        \|d(\cdot,\omega)\|_{\infty,[0,T]}
        \leq \bar d_T(\delta)
    \right\}.
\end{equation}
By Assumption~\ref{ass:uniform-envelope},
$\mathbb{P}(\mathcal{E}_{T,\delta})\geq 1-\delta$.

Fix an arbitrary outcome $\omega\in\mathcal{E}_{T,\delta}$ and an arbitrary
time $\tau\in[0,T]$. Since $[0,\tau]\subseteq[0,T]$,
\begin{equation}
    \|d(\cdot,\omega)\|_{\infty,[0,\tau]}
    \leq
    \|d(\cdot,\omega)\|_{\infty,[0,T]}
    \leq \bar d_T(\delta).
\end{equation}
Applying the deterministic ISS estimate gives
\begin{align}
    \|e(\tau,\omega)\|
    &\leq
    \beta\bigl(\|e(0)\|,\tau\bigr)
    +
    \gamma\left(
        \|d(\cdot,\omega)\|_{\infty,[0,\tau]}
    \right)
    \\
    &\leq
    \beta\bigl(\|e(0)\|,\tau\bigr)
    +
    \gamma\bigl(\bar d_T(\delta)\bigr)
    =
    r(\tau;T,\delta).
\end{align}
Because $\tau\in[0,T]$ was arbitrary, this inequality holds simultaneously
for all $\tau\in[0,T]$ on the event $\mathcal{E}_{T,\delta}$. Hence
\begin{equation}
\begin{aligned}
    \mathcal E_{T,\delta}\subseteq
    \{\omega:{}&\|e(\tau,\omega)\|\leq r(\tau;T,\delta),\\
    &\text{for all }\tau\in[0,T]\}.
\end{aligned}
\end{equation}
which proves \eqref{eq:uniform-hp-tube}. Taking the supremum over
$\tau\in[0,T]$ proves \eqref{eq:constant-uniform-tube}.
\end{proof}

\begin{corollary}[Random initial-error envelope]
\label{cor:random-initial-envelope}
Suppose that Assumption~\ref{ass:iss} holds for every admissible realization. Let $T>0$, $r_0\geq0$, $\bar d_T\geq0$, and let $\delta_0,\delta_d\in(0,1)$ satisfy $\delta_0+\delta_d<1$. If
\begin{align}
\mathbb P\left(\|e(0)\|\leq r_0\right)&\geq1-\delta_0,\\
\mathbb P\left(\|d\|_{\infty,[0,T]}\leq \bar d_T\right)&\geq1-\delta_d,
\end{align}
then
\begin{equation}
\begin{aligned}
\mathbb P\bigl(&\|e(\tau)\|\leq \beta(r_0,\tau)+\gamma(\bar d_T),\\
&\text{for all }\tau\in[0,T]\bigr)
\geq 1-\delta_0-\delta_d .
\end{aligned}
\end{equation}
If a joint event $\{\|e(0)\|\leq r_0,\ \|d\|_{\infty,[0,T]}\leq\bar d_T\}$ is available with probability at least $1-\delta$, the same conclusion holds with $1-\delta$ in place of $1-\delta_0-\delta_d$.
\end{corollary}

\begin{proof}
On the intersection of the two events, monotonicity of $\beta$ in its first argument and monotonicity of $\gamma$ give, for every $\tau\in[0,T]$,
\begin{equation}
\|e(\tau)\|\leq \beta(\|e(0)\|,\tau)+\gamma(\|d\|_{\infty,[0,\tau]})
\leq \beta(r_0,\tau)+\gamma(\bar d_T).
\end{equation}
The probability of this intersection is at least $1-\delta_0-\delta_d$ by the union bound. The joint-event version follows directly from the same deterministic implication.
\end{proof}

\begin{remark}[Time-varying tube]
The function $r(\tau;T,\delta)$ defines a time-varying simultaneous tube.
The scalar $r_T(\delta)$ defines a constant-radius tube and is generally more
conservative. Since $\beta(r,\tau)$ is nonincreasing in its second argument,
\begin{equation}
    r_T(\delta)
    =
    \beta\bigl(\|e(0)\|,0\bigr)
    +
    \gamma\bigl(\bar d_T(\delta)\bigr).
\end{equation}
The time-varying representation is less conservative and is used in the numerical illustrations.
\end{remark}

\begin{remark}[Fixed-time relation]
For any selected $t\in[0,T]$, the fixed-time high-probability estimate follows
immediately from Theorem~\ref{thm:uniform-hp-iss-tube}. The converse requires additional dependence information: separate fixed-time probability statements alone are insufficient for a simultaneous guarantee over an entire interval.
\end{remark}

\begin{corollary}[Fixed-time ISS bound]
\label{cor:fixed-time-hp-tube}
Let $t>0$ and $\delta\in(0,1)$. If Assumption~\ref{ass:iss} holds and there exists
$\bar d_t(\delta)\geq0$ such that
\begin{equation}
    \mathbb{P}\left(
        \|d\|_{\infty,[0,t]}\leq \bar d_t(\delta)
    \right)\geq1-\delta,
\end{equation}
then
\begin{equation}
    \mathbb{P}\left(
        \|e(t)\|
        \leq
        \beta(\|e(0)\|,t)+\gamma(\bar d_t(\delta))
    \right)\geq1-\delta.
\end{equation}
\end{corollary}

\begin{proof}
Apply Theorem~\ref{thm:uniform-hp-iss-tube} with $T=t$ and evaluate the simultaneous bound at the terminal time.
\end{proof}

\subsection{Quadratic Lyapunov tubes}

Theorem~\ref{thm:uniform-hp-iss-tube} reduces probabilistic tube guarantees to verifying a deterministic ISS estimate. We now construct such estimates with quadratic Lyapunov functions, because this form leads directly to ellipsoidal and componentwise tubes and later to LMI-based observer design.

\medskip\noindent\textbf{Aggregated disturbance.}
In what follows, we aggregate process disturbances and measurement noise into a single
input $d(t)$ defined as
\begin{equation}
d(t) := \begin{bmatrix} w(t) \\ v(t) \end{bmatrix},
\end{equation}
equipped with the weighted norm
\begin{equation}
\|d(t)\|^2 := \|w(t)\|^2 + \rho^2 \|v(t)\|^2,
\label{eq:d-weighted-norm}
\end{equation}
where $\rho>0$ is a fixed weighting parameter.

The following proposition is the computational bridge between the abstract ISS bound and the later tube-design LMIs.

\begin{proposition}[Quadratic ISS condition]
\label{prop:lyap-to-iss}
Assume there exists a symmetric matrix $P\succ 0$ and constants $a>0$, $b>0$ such that
the function $V(e)=e^\top P e$ satisfies
\begin{equation}
\begin{aligned}
\dot V(e(t))&\le -a V(e(t)) + b\|d(t)\|^2,\\
&\hspace{6mm}\text{for almost all } t\ge 0,
\end{aligned}
\label{eq:lyap-diss}
\end{equation}
where $\|\cdot\|$ is defined in \eqref{eq:d-weighted-norm}.

Then Assumption~\ref{ass:iss} holds with the explicit functions
\begin{align}
\beta(r,t)&=\sqrt{\frac{\lambda_{\max}(P)}{\lambda_{\min}(P)}}
e^{-at/2}r,\nonumber\\
\gamma(s)&=\sqrt{\frac{b}{a\lambda_{\min}(P)}}s.
\label{eq:beta-gamma}
\end{align}
\end{proposition}

\begin{proof}
Let $V(t):=V(e(t))$. From \eqref{eq:lyap-diss},
\[
\dot V(t) \le -a V(t) + b\|d(t)\|^2.
\]
Multiply both sides by $e^{at}$ and integrate from $0$ to $t$:
\begin{align*}
\frac{d}{dt}\bigl(e^{at}V(t)\bigr)
&\leq b e^{at}\|d(t)\|^2,\\
e^{at}V(t)-V(0)
&\leq b\int_0^t e^{a\tau}\|d(\tau)\|^2\,d\tau.
\end{align*}
Hence,
\begin{align}
V(t) &\le e^{-at}V(0) + b\int_0^t e^{-a(t-\tau)}\|d(\tau)\|^2\,d\tau \nonumber\\
&\le e^{-at}V(0) + b\|d\|_{\infty,[0,t]}^2\int_0^t e^{-a(t-\tau)}\,d\tau \nonumber\\
&= e^{-at}V(0) + \frac{b}{a}\big(1-e^{-at}\big)\|d\|_{\infty,[0,t]}^2 \nonumber\\
&\le e^{-at}V(0) + \frac{b}{a}\|d\|_{\infty,[0,t]}^2.
\label{eq:V-bound}
\end{align}

Using the eigenvalue bounds
\[
\lambda_{\min}(P)\|e(t)\|^2 \le V(t)\le \lambda_{\max}(P)\|e(t)\|^2,
\]
we obtain from \eqref{eq:V-bound} that
\begin{align*}
\lambda_{\min}(P)\|e(t)\|^2
&\le e^{-at}\lambda_{\max}(P)\|e(0)\|^2\\
&\quad+\frac{b}{a}\|d\|_{\infty,[0,t]}^2.
\end{align*}
Taking square roots and using $\sqrt{\alpha+\beta}\le \sqrt{\alpha}+\sqrt{\beta}$ yields
\begin{align*}
\|e(t)\|\le{}&
\sqrt{\frac{\lambda_{\max}(P)}{\lambda_{\min}(P)}}e^{-at/2}\|e(0)\|\\
&+\sqrt{\frac{b}{a\lambda_{\min}(P)}}\|d\|_{\infty,[0,t]}.
\end{align*}
This is precisely \eqref{eq:iss-bound} with $\beta,\gamma$ given by \eqref{eq:beta-gamma}.
\end{proof}

The next corollary applies the finite-horizon lifting to this quadratic estimate and gives the tube radii used in the numerical sections.

\begin{corollary}[Quadratic high-probability tube]
\label{cor:uniform-quadratic-tube}
Suppose that the conditions of Proposition~\ref{prop:lyap-to-iss} hold and that
Assumption~\ref{ass:uniform-envelope} is satisfied. Define
\begin{align}
    \rho_V(\tau;T,\delta)
    :={}&e^{-a\tau}V(e(0))\\
    &+\frac{b}{a}\bigl(1-e^{-a\tau}\bigr)\bar d_T^2(\delta),
    \quad \tau\in[0,T].
    \label{eq:uniform-rho-v}
\end{align}
Then
\begin{equation}
\begin{aligned}
\mathbb{P}\bigl(&e(\tau)^\top Pe(\tau)\leq\rho_V(\tau;T,\delta),\\
&\text{for all }\tau\in[0,T]\bigr)\geq1-\delta.
\end{aligned}
\label{eq:uniform-ellipsoidal-tube}
\end{equation}
Consequently,
\begin{equation}
\begin{aligned}
\mathbb{P}\biggl(&\|e(\tau)\|\leq
\sqrt{\frac{\rho_V(\tau;T,\delta)}{\lambda_{\min}(P)}},\\
&\text{for all }\tau\in[0,T]\biggr)\geq1-\delta.
\end{aligned}
\label{eq:uniform-euclidean-quadratic-tube}
\end{equation}
Moreover, the same horizon-level event gives the joint componentwise box statement
\begin{equation}
\begin{aligned}
\mathbb{P}\bigl(&|e_i(\tau)|\leq
\sqrt{(P^{-1})_{ii}\rho_V(\tau;T,\delta)},\\
&i=1,\ldots,n,\ \tau\in[0,T]\bigr)\geq1-\delta.
\end{aligned}
\label{eq:uniform-componentwise-tube}
\end{equation}
Each individual componentwise statement follows immediately, but no additional union bound over the state components is required.
\end{corollary}

\begin{proof}
Multiplying the differential inequality
\[
    \dot V(t)\leq -aV(t)+b\|d(t)\|^2
\]
by $e^{at}$ and integrating from $0$ to $\tau$ gives
\begin{equation}
    V(\tau)
    \leq
    e^{-a\tau}V(0)
    +
    b\int_0^\tau
        e^{-a(\tau-s)}\|d(s)\|^2\,ds.
\end{equation}
On the event
$\{\|d\|_{\infty,[0,T]}\leq\bar d_T(\delta)\}$, this yields, simultaneously
for every $\tau\in[0,T]$,
\begin{align}
    V(\tau)
    &\leq
    e^{-a\tau}V(0)
    +
    b\bar d_T^2(\delta)
    \int_0^\tau e^{-a(\tau-s)}\,ds
    \\
    &=
    e^{-a\tau}V(0)
    +
    \frac{b}{a}
    \left(1-e^{-a\tau}\right)
    \bar d_T^2(\delta).
\end{align}
Equation~\eqref{eq:uniform-euclidean-quadratic-tube} follows from
$e^\top Pe\geq\lambda_{\min}(P)\|e\|^2$.

For the componentwise result, the maximum of $e_i^2$ over the ellipsoid
$\{e:e^\top Pe\leq\rho\}$ is
$\rho(P^{-1})_{ii}$. Substitution of
$\rho=\rho_V(\tau;T,\delta)$ gives
\eqref{eq:uniform-componentwise-tube}.
\end{proof}

For a constant-radius version, define
\begin{equation}
    \rho_{V,T}(\delta)
    :=
    \max_{\tau\in[0,T]}
    \rho_V(\tau;T,\delta).
    \label{eq:rho-v-horizon}
\end{equation}
Then
\begin{equation}
    \mathbb{P}\left(
        \sup_{\tau\in[0,T]}\|e(\tau)\|
        \leq
        \sqrt{
            \frac{\rho_{V,T}(\delta)}
                 {\lambda_{\min}(P)}
        }
    \right)
    \geq 1-\delta .
    \label{eq:uniform-quadratic-constant-radius}
\end{equation}
Writing
\begin{align*}
S_T(\delta)&:=\frac{b}{a}\bar d_T^2(\delta),\\
R_T(\delta)&:=e^{-aT}V(e(0))+
\bigl(1-e^{-aT}\bigr)S_T(\delta),
\end{align*}
the finite-horizon maximum is
\begin{equation}
\rho_{V,T}(\delta)=
\begin{cases}
V(e(0)),&V(e(0))\geq S_T(\delta),\\
R_T(\delta),&V(e(0))<S_T(\delta).
\end{cases}
\label{eq:explicit-rho-v-horizon}
\end{equation}

\section{Observer contraction conditions in constant metrics}
\label{sec:observer-conditions}

The preceding results require a deterministic ISS inequality or, more specifically, a quadratic dissipation inequality. This section records two sufficient constructions. The first is a general incremental-dissipativity statement. The second exploits a known nonlinear channel and is the condition used in the flexible-joint study.

\subsection{A conservative full-norm sufficient condition}

Consider the plant--observer pair
\begin{align}
\dot x(t) &= f(x(t),u(t),t) + w(t), \label{eq:plant}\\
y(t) &= h(x(t),t) + v(t), \label{eq:meas}\\
\dot{\hat x}(t) &= f(\hat x(t),u(t),t) + L\big(y(t)-h(\hat x(t),t)\big), \label{eq:observer}
\end{align}
where $x(t),\hat x(t)\in\mathbb{R}^n$, $y(t)\in\mathbb{R}^m$, and $w(t),v(t)$ are
measurable disturbance/noise signals. Define the estimation error $e(t)=x(t)-\hat x(t)$.
Subtracting \eqref{eq:observer} from \eqref{eq:plant} and using \eqref{eq:meas} yields
\begin{align}
\dot e(t)
={}&f(x(t),u(t),t)-f(\hat x(t),u(t),t)\nonumber\\
&-L\bigl(h(x(t),t)-h(\hat x(t),t)\bigr)\nonumber\\
&+w(t)-Lv(t).
\label{eq:error-dynamics}
\end{align}

The following condition is included to connect the abstract ISS assumption with a directly checkable nonlinear inequality. This is a classical quadratic route to ISS and is used here as a sufficient condition linking the abstract assumption to a checkable inequality. It is deliberately conservative: the unforced plant flow must already be incrementally dissipative in the selected metric, while the output injection is bounded in magnitude and can only consume part of that margin. The flexible-joint design later uses the structured condition in Section~\ref{sec:structured-contraction}, which exploits the known nonlinear channel.

\begin{corollary}[Constant-metric ISS condition]
\label{cor:incremental-full}
Assume that there exist a symmetric matrix $P\succ 0$, constants $\lambda>0$ and
$L_h\ge 0$ such that for all $x,\hat x\in\mathbb{R}^n$ and all $t$:
\begin{enumerate}
\item[(A1)] (Incremental dissipativity)
\begin{equation}
(x-\hat x)^\top P\big(f(x,u,t)-f(\hat x,u,t)\big)
\le -\lambda \|x-\hat x\|^2 .
\label{eq:A1}
\end{equation}
\item[(A2)] (Lipschitz output map)
\begin{equation}
\|h(x,t)-h(\hat x,t)\| \le L_h \|x-\hat x\| .
\label{eq:A2}
\end{equation}
\item[(A3)] (Gain compatibility)
\begin{equation}
\|PL\|\,L_h \le \lambda/2 ,
\label{eq:A3}
\end{equation}
where $\|\cdot\|$ denotes the induced operator norm consistent with the vector norm
$\|\cdot\|$.
\end{enumerate}
Then the estimation error system \eqref{eq:error-dynamics} is input-to-state stable with
respect to the aggregated disturbance $d(t)=(w(t),v(t))$ in the following explicit sense:
there exist $\beta\in\mathcal{KL}$ and $\gamma\in\mathcal{K}$ such that for all $t\ge 0$,
\[
\|e(t)\| \le \beta(\|e(0)\|,t) + \gamma\big(\|w\|_{\infty,[0,t]}+\|v\|_{\infty,[0,t]}\big).
\]
\end{corollary}

\begin{proof}
Let $V=e^\top Pe$. Conditions (A1)--(A3) imply that the nominal drift and output-injection terms satisfy
\begin{align}
2e^\top P\bigl(f(x,u,t)-f(\hat x,u,t)\bigr)
&\nonumber\\[-1mm]
-2e^\top PL\bigl(h(x,t)-h(\hat x,t)\bigr)
&\leq-\lambda\|e\|^2.
\end{align}
Applying Young's inequality to $2e^\top Pw-2e^\top PLv$, with both free parameters equal to $\lambda/4$, gives
\begin{align}
\dot V\leq{}&-\frac{\lambda}{2}\|e\|^2
+\frac{4\|P\|^2}{\lambda}\|w\|^2\\
&+\frac{4\|PL\|^2}{\lambda}\|v\|^2.
\end{align}
With $\rho=\|PL\|/\|P\|$ and $\|d\|^2=\|w\|^2+\rho^2\|v\|^2$, this becomes
\begin{align}
\dot V&\leq-aV+b\|d\|^2,\\
a&=\frac{\lambda}{2\lambda_{\max}(P)},\qquad
b=\frac{4\|P\|^2}{\lambda}.
\end{align}
Proposition~\ref{prop:lyap-to-iss} then yields the stated ISS estimate.
\end{proof}

\medskip
Corollary~\ref{cor:incremental-full} is a sufficient bridge from full-norm incremental inequalities to the abstract ISS assumption. Its role is expository; failure of this condition has no implication for structured, local, state-dependent, or nonquadratic observer constructions.

\subsection{Structured LMI condition for nonlinear observer contraction}
\label{sec:structured-contraction}

Quadratic linear matrix inequality (LMI) conditions for Lipschitz nonlinear observers are well established. Rajamani gave a foundational constant-metric construction \cite{rajamani1998observers}. Abbaszadeh and Marquez developed a sampled-data $\mathcal H_\infty$ design \cite{abbaszadeh2008sampled}; Zemouche et al. derived enhanced circle-criterion conditions \cite{zemouche2017circle}; and Beikzadeh and Marquez formulated an input-to-error stable sampled-data observer \cite{beikzadeh2016input}. The condition below is a direct structured specialization for the model used later. Its purpose is to retain the known location $B_\phi$ and argument $H_\phi x$ of the nonlinearity and to provide an affine inner constraint for the finite-horizon tube co-design. Its role is a problem-specific sufficient condition for the subsequent tube co-design.

Consider the nonlinear system
\begin{align}
\dot x(t)={}&Ax(t)+Bu(t)\nonumber\\
&+B_\phi\phi(H_\phi x(t))+w(t),
\label{eq:structured-plant}\\
y(t)={}&Cx(t)+v(t).
\label{eq:structured-output}
\end{align}
where $B_\phi$ locates the nonlinear term and $H_\phi$ selects its argument. The observer is
\begin{align}
\dot{\hat x}(t)={}&A\hat x(t)+Bu(t)+B_\phi\phi(H_\phi\hat x(t))\nonumber\\
&+L\bigl(y(t)-C\hat x(t)\bigr).
\label{eq:structured-observer}
\end{align}
For $e=x-\hat x$, define
\begin{equation}
\Delta\phi=\phi(H_\phi x)-\phi(H_\phi\hat x).
\end{equation}
Assume that the scalar or vector nonlinearity satisfies the incremental bound
\begin{equation}
\|\Delta\phi\|\leq \gamma_\phi\|H_\phi e\|.
\label{eq:structured-increment}
\end{equation}
This condition preserves the direction $B_\phi$ and the argument $H_\phi e$, while avoiding a full-state norm bound for the nonlinear term. The following proposition states the corresponding constant-metric contraction test.

\begin{proposition}[Structured contraction condition]
\label{prop:structured-contraction}
Let $\alpha>0$ be prescribed. Suppose there exist $P\succ0$, $Y$, and $\tau\geq0$ such that
\begin{equation}
\begin{bmatrix}
\operatorname{He}(PA-YC)+\alpha P
+\tau\gamma_\phi^2H_\phi^\top H_\phi
& PB_\phi\\
B_\phi^\top P & -\tau I
\end{bmatrix}
\preceq0.
\label{eq:structured-contraction-lmi}
\end{equation}
Then $L=P^{-1}Y$ satisfies
\begin{equation}
2e^\top P\bigl((A-LC)e+B_\phi\Delta\phi\bigr)
\leq-\alpha e^\top Pe
\label{eq:structured-contraction-inequality}
\end{equation}
for every pair $(e,\Delta\phi)$ satisfying \eqref{eq:structured-increment}.
\end{proposition}

\begin{proof}
Let $z=[e^\top\ \Delta\phi^\top]^\top$. The incremental condition is equivalent to
\begin{equation}
z^\top
\begin{bmatrix}
\gamma_\phi^2H_\phi^\top H_\phi&0\\
0&-I
\end{bmatrix}z\geq0.
\end{equation}
The quadratic form associated with \eqref{eq:structured-contraction-inequality}, augmented by $\tau$ times this constraint, is precisely \eqref{eq:structured-contraction-lmi}. The S-procedure therefore gives the stated implication.
\end{proof}

For an additive disturbance $d_h$ in the error dynamics, \eqref{eq:structured-contraction-inequality} gives
\begin{equation}
\dot V\leq-\alpha V+2e^\top Pd_h.
\end{equation}
For any $\varepsilon\in(0,\alpha)$, Young's inequality yields
\begin{equation}
\dot V\leq-(\alpha-\varepsilon)V
+\frac{\lambda_{\max}(P)}{\varepsilon}\|d_h\|^2.
\label{eq:structured-tube-constants-derivation}
\end{equation}
Hence one admissible pair of tube constants is
\begin{equation}
a_h=\alpha-\varepsilon,
\qquad
b_h=\frac{\lambda_{\max}(P)}{\varepsilon}.
\label{eq:structured-tube-constants}
\end{equation}
The contraction-feasibility study and the tube-width co-design may select different rates: $\alpha$ measures a feasible nominal contraction margin, whereas the decay rate $a$ in the co-design balances convergence against disturbance amplification.

\subsection{Diagonal-metric specialization and positive-system motivation}
\label{sec:positive-subclass}

The following result is a diagonal-metric specialization of the general quadratic construction. Its proof is independent of positivity, while the motivation for recording it separately comes from positive and compartmental models.

\begin{corollary}[Diagonal-metric case]
\label{cor:diagonal-metric}
Consider
\begin{equation}
\dot e=Fe+Gd.
\end{equation}
Suppose that there exist a diagonal matrix $P=\operatorname{diag}(p_1,\ldots,p_n)\succ0$ and a scalar $q>0$ such that
\begin{equation}
\operatorname{He}(PF)\preceq-qI.
\label{eq:positive-diagonal-condition}
\end{equation}
Then the error dynamics satisfy
\begin{equation}
\dot V\leq-aV+b\|d\|_2^2,
\qquad
V=e^\top Pe,
\label{eq:positive-diagonal-dissipation}
\end{equation}
with
\begin{equation}
a=\frac{q}{2\lambda_{\max}(P)},
\qquad
b=\frac{2\|PG\|_2^2}{q}.
\end{equation}
Consequently, every deterministic or finite-horizon probabilistic envelope for $d$ yields the tubes of Corollary~\ref{cor:uniform-quadratic-tube}, and their componentwise half-widths simplify to $\sqrt{\rho_V(t)/p_i}$.
\end{corollary}

\begin{proof}
Along the error dynamics,
\begin{equation}
\dot V=e^\top\operatorname{He}(PF)e+2e^\top PGd.
\end{equation}
Condition~\eqref{eq:positive-diagonal-condition} and Young's inequality give
\begin{align}
\dot V
&\leq-q\|e\|_2^2+2\|PG\|_2\|e\|_2\|d\|_2\\
&\leq-\frac{q}{2}\|e\|_2^2+\frac{2\|PG\|_2^2}{q}\|d\|_2^2.
\end{align}
Since $\|e\|_2^2\geq V/\lambda_{\max}(P)$, inequality~\eqref{eq:positive-diagonal-dissipation} follows. The tube statements then follow from Proposition~\ref{prop:lyap-to-iss} and Corollary~\ref{cor:uniform-quadratic-tube}. For diagonal $P$, $(P^{-1})_{ii}=1/p_i$.
\end{proof}

Positive and cooperative systems motivate this specialization because nonnegative states have direct physical meaning and diagonal Lyapunov functions are often natural. These properties are developed systematically in the monographs by Farina and Rinaldi \cite{farina2000positive} and by Kaczorek \cite{kaczorek2002positive}. Positive error dynamics also underpin interval observers. Zhang and Shen survey this design approach \cite{zhang2023intervalsurvey}; Dinh and Tran give a systematic linear construction \cite{dinh2025interval}; and Thabet et al. develop a reduced-order design with unknown inputs and bounded disturbances \cite{thabet2026interval}. In the present construction, the estimation error may be signed and the method propagates one tube around a nominal estimate, while interval observers propagate lower and upper trajectories. In the compartment benchmark the nominal matrix is Metzler and the output matrix is nonnegative; signed additive disturbances preserve state positivity only under additional cone-preserving assumptions. Positivity constraints in MHE are therefore a separate estimator-side choice.

\section{Sampled-data tube and observer co-design}
\label{sec:codesign}

This section connects the preceding ISS and contraction results to the sampled-data observer used in the numerical study. The progression is as follows. First, the held measurement is written as the current output plus a staleness error. Second, process disturbance, sampled measurement noise, and held-output mismatch are collected into a composite disturbance channel. Third, the sampled-data tube theorem is proved. Finally, this tube is turned into a design objective, first for a linear observer and then for a structured nonlinear observer.

\subsection{Sampled-data observer with held measurements}
\label{sec:sampled-data-observer}

The previous sections provide the contraction and tube machinery. We now apply it to the actual sampled-data implementation. The central modelling point is to keep process disturbance, held measurement noise, and intersample output mismatch as separate channels, because the observer gain affects them differently.

The continuous-measurement formulation represents an idealized access pattern. In many implementations the plant evolves continuously while output measurements are available only at discrete sampling instants. We therefore consider a sampling sequence
\begin{align}
    0=t_0<t_1<t_2<\cdots,\qquad
    h_k&:=t_{k+1}-t_k,\\
    0<h_k&\leq \bar h.
    \label{eq:sampling-sequence}
\end{align}
where $\bar h$ is the maximum admissible sampling interval.

The plant evolves continuously according to
\begin{equation}
    \dot x(t)
    =
    f(x(t),u(t),t)+w(t),
    \label{eq:sampled-plant}
\end{equation}
whereas the measured output is available only at $t_k$:
\begin{equation}
    y_k
    =
    h(x(t_k),t_k)+v_k.
    \label{eq:sampled-output}
\end{equation}
Here, $v_k\in\mathbb{R}^m$ denotes measurement noise at the $k$th sampling
instant.

Define the sampling-time map
\begin{equation}
    \sigma(t)
    :=
    t_k,
    \qquad
    t\in[t_k,t_{k+1}),
    \label{eq:sampling-time-map}
\end{equation}
and the zero-order-held signals
\begin{align}
    \bar y(t)&:=y_k,\qquad \bar v(t):=v_k,\\
    t&\in[t_k,t_{k+1}).
    \label{eq:held-signals}
\end{align}
The continuous-time observer driven by the held measurement is
\begin{align}
    \dot{\hat x}(t)
    ={}&f(\hat x(t),u(t),t)\\
    &+L\bigl(\bar y(t)-h(\hat x(t),t)\bigr).
    \label{eq:sampled-data-observer}
\end{align}

\medskip\noindent\textbf{Intersample output mismatch.}
The held measurement contains information about the state at $\sigma(t)$, while the current plant state is evaluated at $t$. Define
\begin{equation}
    \zeta_h(t)
    :=
    h(x(t),t)
    -
    h(x(\sigma(t)),\sigma(t)).
    \label{eq:intersample-output-mismatch}
\end{equation}
Using
\begin{align}
    \bar y(t)
    &=h(x(\sigma(t)),\sigma(t))+\bar v(t)\\
    &=h(x(t),t)-\zeta_h(t)+\bar v(t).
\end{align}
the estimation error $e(t)=x(t)-\hat x(t)$ satisfies
\begin{align}
    \dot e(t)
    ={}&
    f(x(t),u(t),t)-f(\hat x(t),u(t),t)
    \nonumber\\
    &-
    L\left(
        h(x(t),t)-h(\hat x(t),t)
    \right)
    +
    d_h(t),
    \label{eq:sampled-error-dynamics}
\end{align}
where
\begin{equation}
    d_h(t)
    :=
    w(t)-L\bar v(t)+L\zeta_h(t).
    \label{eq:sampled-aggregate-disturbance}
\end{equation}
Thus, sampling introduces an additional disturbance term $L\zeta_h(t)$.
This term vanishes in the continuous-measurement limit $\bar h\to 0$.

\begin{assumption}[Contraction condition]
\label{ass:sampled-nominal-contraction}
There exist a symmetric matrix $P\succ0$ and a constant $\alpha>0$ such that
for all admissible $x,\hat x,u,t$, with $e=x-\hat x$,
\begin{align}
    2e^\top P\bigl(f(x,u,t)-f(\hat x,u,t)\bigr)
    &\nonumber\\
    {}-2e^\top PL\bigl(h(x,t)-h(\hat x,t)\bigr)
    &\leq-\alpha e^\top Pe.
    \label{eq:sampled-contraction-condition}
\end{align}
\end{assumption}

Assumption~\ref{ass:sampled-nominal-contraction} is the nominal contraction requirement for the observer before the sampled-data disturbance channels are added. It is restrictive in the sense that it asks for one constant quadratic metric on the selected operating region. The following sections make this requirement checkable through linear and structured nonlinear LMIs, so it is used here as a transparent intermediate condition.

\begin{assumption}[Output variation]
\label{ass:intersample-output-variation}
There exists a nondecreasing function
$\omega_h:\mathbb{R}_{\geq0}\rightarrow\mathbb{R}_{\geq0}$ satisfying
$\omega_h(0)=0$ such that
\begin{equation}
    \|\zeta_h(t)\|
    \leq
    \omega_h(\bar h)
    \qquad
    \text{for all }t\in[0,T].
    \label{eq:intersample-output-bound}
\end{equation}
\end{assumption}

Assumption~\ref{ass:intersample-output-variation} says that the measured output cannot change arbitrarily fast between two sampling instants. This is mild on compact operating regions for locally Lipschitz output maps and bounded state velocities. It becomes restrictive when the state can leave the prescribed region, or when unmodelled fast output dynamics are present. A directly verifiable sufficient condition is given next.

The following lemma supplies the staleness bound used in the sampled-data tube.

\begin{lemma}[Output-staleness bound]
\label{lem:output-staleness}
Suppose that, on the considered forward-invariant region,
\begin{align}
    \|h(x_1,t)-h(x_2,t)\|
    &\leq
    L_{h,x}\|x_1-x_2\|,
    \label{eq:h-state-lipschitz}\\
    \|h(x,t_1)-h(x,t_2)\|
    &\leq
    L_{h,t}|t_1-t_2|,
    \label{eq:h-time-lipschitz}
\end{align}
and that
\begin{equation}
    \|\dot x(t)\|\leq M_x
    \qquad
    \text{for all }t\in[0,T].
    \label{eq:state-rate-bound}
\end{equation}
Then Assumption~\ref{ass:intersample-output-variation} holds with
\begin{equation}
    \omega_h(\bar h)
    =
    \left(
        L_{h,x}M_x+L_{h,t}
    \right)\bar h.
    \label{eq:explicit-staleness-bound}
\end{equation}
\end{lemma}

\begin{proof}
For $t\in[t_k,t_{k+1})$, add and subtract $h(x(t_k),t)$ to obtain
\begin{align}
    \|\zeta_h(t)\|
    \leq{}&
    \|h(x(t),t)-h(x(t_k),t)\|
    \nonumber\\
    &+
    \|h(x(t_k),t)-h(x(t_k),t_k)\|.
\end{align}
The Lipschitz assumptions give
\begin{equation}
    \|\zeta_h(t)\|
    \leq
    L_{h,x}\|x(t)-x(t_k)\|
    +
    L_{h,t}(t-t_k).
\end{equation}
Moreover,
\begin{equation}
    \|x(t)-x(t_k)\|
    \leq
    \int_{t_k}^{t}\|\dot x(s)\|\,ds
    \leq
    M_x(t-t_k).
\end{equation}
Since $t-t_k\leq h_k\leq\bar h$, the result follows.
\end{proof}

\begin{assumption}[Disturbance envelopes]
\label{ass:sampled-probabilistic-envelopes}
Let $T>0$. For risk allocations
$\delta_w,\delta_v\in(0,1)$, assume that deterministic quantities
$\bar w_T(\delta_w)$ and $\bar v_T(\delta_v)$ satisfy
\begin{align}
    \mathbb{P}\left(
        \|w\|_{\infty,[0,T]}
        \leq
        \bar w_T(\delta_w)
    \right)
    &\geq
    1-\delta_w,
    \label{eq:process-envelope}\\
    \mathbb{P}\left(
        \max_{0\leq k\leq N_T}\|v_k\|
        \leq
        \bar v_T(\delta_v)
    \right)
    &\geq
    1-\delta_v,
    \label{eq:sample-noise-envelope}
\end{align}
where
\begin{equation}
    N_T
    :=
    \max\{k:t_k\leq T\}.
\end{equation}
No independence between $w$ and $(v_k)$ is required.
\end{assumption}

Assumption~\ref{ass:sampled-probabilistic-envelopes} assigns a finite-horizon risk budget to the input channels. In words, the process disturbance and all sampled measurement errors over $[0,T]$ are probably bounded by deterministic numbers that depend on the horizon, the sampling grid, and the selected failure probabilities. The statement allows conservative constructions, such as union bounds over all samples, and does not require independent process and measurement disturbances.

Because the noise is held between sampling instants,
\begin{equation}
    \|\bar v\|_{\infty,[0,T]}
    =
    \max_{0\leq k\leq N_T}\|v_k\|.
    \label{eq:held-noise-maximum}
\end{equation}

Combining the nominal contraction condition, the staleness bound, and the finite-horizon disturbance envelopes gives the sampled-data tube theorem.

\begin{theorem}[Sampled-data tube]
\label{thm:sampled-data-uniform-tube}
Suppose that Assumptions~\ref{ass:sampled-nominal-contraction},
\ref{ass:intersample-output-variation}, and
\ref{ass:sampled-probabilistic-envelopes} hold. Let
\begin{equation}
    \delta_w+\delta_v\leq\delta
\end{equation}
and define
\begin{align}
    \bar d_{h,T}(\delta)
    :={}&\bar w_T(\delta_w)\\
    &+\|L\|\bigl(\bar v_T(\delta_v)+\omega_h(\bar h)\bigr).
    \label{eq:sampled-disturbance-envelope}
\end{align}
For any $\varepsilon\in(0,\alpha)$, let
\begin{equation}
    a_h:=\alpha-\varepsilon,
    \qquad
    b_h:=\frac{\lambda_{\max}(P)}{\varepsilon}.
    \label{eq:sampled-lyapunov-constants}
\end{equation}
Then, with
\begin{align}
    \rho_h(\tau;T,\delta)
    :={}&e^{-a_h\tau}V(e(0))\\
    &+\frac{b_h}{a_h}\bigl(1-e^{-a_h\tau}\bigr)
      \bar d_{h,T}^2(\delta).
    \label{eq:sampled-rho}
\end{align}
the following simultaneous guarantee holds:
\begin{equation}
\begin{aligned}
    \mathbb{P}\Bigl(&e(\tau)^\top Pe(\tau)
      \leq\rho_h(\tau;T,\delta)\\
      &\text{for all }\tau\in[0,T]\Bigr)
    \geq1-\delta.
\end{aligned}
    \label{eq:sampled-ellipsoidal-tube}
\end{equation}
Consequently,
\begin{equation}
\begin{aligned}
    \mathbb{P}\Bigl(&\|e(\tau)\|
      \leq\sqrt{\rho_h(\tau;T,\delta)/\lambda_{\min}(P)}\\
      &\text{for all }\tau\in[0,T]\Bigr)
    \geq1-\delta.
\end{aligned}
    \label{eq:sampled-euclidean-tube}
\end{equation}
and the joint componentwise statement
\begin{align}
    r_{i,h}(\tau)&:=
      \sqrt{(P^{-1})_{ii}\rho_h(\tau;T,\delta)},\\
    \mathbb{P}\Bigl(&|e_i(\tau)|\leq r_{i,h}(\tau),\nonumber\\
      &i=1,\ldots,n,\ \tau\in[0,T]\Bigr)\geq1-\delta.
    \label{eq:sampled-componentwise-tube}
\end{align}
\end{theorem}

\begin{proof}
Consider the event
\begin{equation}
\begin{aligned}
    \mathcal{E}_{T,\delta}:=\Bigl\{&
      \|w\|_{\infty,[0,T]}\leq\bar w_T(\delta_w),\\
      &\max_{0\leq k\leq N_T}\|v_k\|
      \leq\bar v_T(\delta_v)\Bigr\}.
\end{aligned}
\end{equation}
The union bound gives
\begin{equation}
    \mathbb{P}(\mathcal{E}_{T,\delta})
    \geq
    1-\delta_w-\delta_v
    \geq
    1-\delta.
\end{equation}
On this event, equations~\eqref{eq:sampled-aggregate-disturbance},
\eqref{eq:intersample-output-bound}, and
\eqref{eq:held-noise-maximum} imply
\begin{equation}
    \|d_h(t)\|
    \leq
    \bar d_{h,T}(\delta)
    \qquad
    \text{for all }t\in[0,T].
\end{equation}

For $V(e)=e^\top Pe$, Assumption~\ref{ass:sampled-nominal-contraction}
gives
\begin{equation}
    \dot V
    \leq
    -\alpha V
    +
    2e^\top P d_h.
\end{equation}
Using Young's inequality in the $P$ metric,
\begin{align}
    2e^\top P d_h
    &\leq
    \varepsilon e^\top Pe
    +
    \frac{1}{\varepsilon}d_h^\top P d_h\\
    &\leq
    \varepsilon V
    +
    \frac{\lambda_{\max}(P)}{\varepsilon}
    \|d_h\|^2.
\end{align}
Therefore,
\begin{equation}
    \dot V
    \leq
    -a_hV+b_h\|d_h\|^2.
\end{equation}
Integration of this scalar differential inequality yields
\begin{equation}
    V(e(\tau))
    \leq
    e^{-a_h\tau}V(e(0))
    +
    \frac{b_h}{a_h}
    \left(
        1-e^{-a_h\tau}
    \right)
    \bar d_{h,T}^2(\delta)
\end{equation}
simultaneously for all $\tau\in[0,T]$ on
$\mathcal{E}_{T,\delta}$. The Euclidean and componentwise bounds follow from
\begin{equation}
    e^\top Pe
    \geq
    \lambda_{\min}(P)\|e\|^2
\end{equation}
and
\begin{equation}
    \max_{e^\top Pe\leq\rho} e_i^2
    =
    \rho(P^{-1})_{ii},
\end{equation}
respectively.
\end{proof}

\begin{remark}[Continuous-measurement limit]
If the sampling interval tends to zero and the output trajectory is uniformly
continuous, then $\omega_h(\bar h)\to0$. The sampled-data tube therefore
converges to its continuous-measurement counterpart, subject to the selected
measurement-noise envelope.
\end{remark}

\begin{remark}[Sampling-frequency trade-off]
Reducing $\bar h$ decreases the deterministic staleness contribution
$\omega_h(\bar h)$. For unbounded measurement-noise distributions, however,
the high-probability envelope for
$\max_{0\leq k\leq N_T}\|v_k\|$ generally increases with the number of samples.
The resulting tube therefore exposes a trade-off between intersample model
error and finite-horizon extreme measurement noise.
\end{remark}

\subsection{Linear sampled-data specialization and tube-oriented design}
\label{sec:linear-sampled-specialization}

The linear case is developed in two steps. The first step specializes the sampled-data error equation and gives the composite disturbance envelope. The second step turns this envelope into the LMI co-design problem. The details are placed in one subsection because the later numerical values are obtained by applying this exact design problem to the compartment benchmark.

Consider
\begin{equation}
    \dot x(t)=Ax(t)+w(t),
    \qquad
    y_k=Cx(t_k)+v_k,
    \label{eq:linear-sampled-plant}
\end{equation}
and the held-measurement observer
\begin{equation}
    \dot{\hat x}(t)
    =
    A\hat x(t)
    +
    L\left(
        \bar y(t)-C\hat x(t)
    \right).
    \label{eq:linear-held-observer}
\end{equation}
Let
\begin{equation}
    F:=A-LC.
\end{equation}
The estimation error then satisfies
\begin{align}
    \dot e(t)={}&Fe(t)+w(t)-L\bar v(t)\\
    &+LC\bigl(x(t)-x(\sigma(t))\bigr).
    \label{eq:linear-sampled-error}
\end{align}

Suppose that $F$ is Hurwitz and let $P\succ0$ solve
\begin{equation}
    F^\top P+PF=-Q,
    \qquad Q\succ0.
    \label{eq:linear-sampled-lyapunov}
\end{equation}
Then the nominal error dynamics satisfy
\begin{equation}
    2e^\top PFe
    =
    -e^\top Qe
    \leq
    -
    \frac{\lambda_{\min}(Q)}
         {\lambda_{\max}(P)}
    e^\top Pe.
\end{equation}
Thus Assumption~\ref{ass:sampled-nominal-contraction} holds with
\begin{equation}
    \alpha
    =
    \frac{\lambda_{\min}(Q)}
         {\lambda_{\max}(P)}.
    \label{eq:linear-alpha}
\end{equation}

If
\begin{equation}
    \|\dot x(t)\|\leq M_x
    \qquad
    \text{on }[0,T],
\end{equation}
then
\begin{equation}
    \|x(t)-x(\sigma(t))\|
    \leq
    M_x\bar h,
\end{equation}
and hence
\begin{equation}
    \|d_h(t)\|
    \leq
    \|w(t)\|
    +
    \|L\|\|\bar v(t)\|
    +
    \|LC\|M_x\bar h.
    \label{eq:linear-sampled-disturbance-bound}
\end{equation}
A finite-horizon high-probability envelope is therefore
\begin{align}
    \bar d_{h,T}(\delta)
    ={}&\bar w_T(\delta_w)+\|L\|\bar v_T(\delta_v)\\
    &+\|LC\|M_x\bar h,\qquad
      \delta_w+\delta_v\leq\delta.
    \label{eq:linear-sampled-envelope}
\end{align}
Substitution of \eqref{eq:linear-sampled-envelope} into
Theorem~\ref{thm:sampled-data-uniform-tube} gives the required sampled-data
tube for the compartment model.

\subsubsection{Tube-oriented linear observer design}
\label{sec:tube-oriented-observer-design}

The ISS tube can be used both as an a posteriori robustness diagnostic and as an observer-design criterion. We formulate a finite-horizon design problem that jointly selects the observer gain and the quadratic Lyapunov metric.

Consider the sampled-data error dynamics
\begin{equation}
    \dot e(t)
    =
    (A-LC)e(t)
    +
    w(t)
    -
    L\bar v(t)
    +
    LCq_h(t),
    \label{eq:design-error-dynamics}
\end{equation}
where
\begin{equation}
    q_h(t)
    :=
    x(t)-x(\sigma(t))
\end{equation}
is the intersample state mismatch.

Let $T>0$ and $\delta\in(0,1)$ be the design horizon and the required risk
level. Suppose that the following finite-horizon envelopes are available:
\begin{align}
    \|w\|_{\infty,[0,T]}
    &\leq
    \bar w_T,
    \label{eq:design-w-envelope}\\
    \|\bar v\|_{\infty,[0,T]}
    &\leq
    \bar v_T(\delta),
    \label{eq:design-v-envelope}\\
    \|q_h\|_{\infty,[0,T]}
    &\leq
    \bar q_h,
    \label{eq:design-q-envelope}
\end{align}
on an event whose probability is at least $1-\delta$. A typical choice for
the intersample envelope is
\begin{equation}
    \bar q_h=M_x\bar h,
    \label{eq:design-q-bound}
\end{equation}
where $\bar h$ is the maximum sampling interval and
$\|\dot x(t)\|\leq M_x$ on the considered region.

Define normalized disturbance signals
\begin{align}
    w(t)&=\bar w_T z_w(t),\\
    \bar v(t)&=\bar v_T(\delta)z_v(t),\\
    q_h(t)&=\bar q_h z_h(t).
\end{align}
where
\begin{equation}
    \|z_w(t)\|,\ \|z_v(t)\|,\ \|z_h(t)\|\leq1.
\end{equation}
Let
\begin{equation}
    z(t)
    :=
    \frac{1}{\sqrt{3}}
    \begin{bmatrix}
        z_w(t)\\
        z_v(t)\\
        z_h(t)
    \end{bmatrix}.
\end{equation}
Then $\|z(t)\|\leq1$ and the error dynamics can be written as
\begin{equation}
    \dot e(t)
    =
    (A-LC)e(t)+G(L)z(t),
    \label{eq:normalized-design-dynamics}
\end{equation}
where
\begin{equation}
    G(L)
    =
    \sqrt{3}
    \begin{bmatrix}
        \bar w_T I
        &
        -\bar v_T(\delta)L
        &
        \bar q_h LC
    \end{bmatrix}.
    \label{eq:design-input-matrix}
\end{equation}

\subsubsection{Initial-error set}

The observer must be designed relative to a prescribed initial-error set.
Let
\begin{equation}
    \mathcal E_0
    :=
    \left\{
        e_0:
        e_0^\top W_0e_0\leq1
    \right\},
    \qquad
    W_0\succ0.
    \label{eq:initial-error-set}
\end{equation}
For example, the Euclidean uncertainty set
$\|e_0\|\leq r_0$ is obtained using
\begin{equation}
    W_0=r_0^{-2}I.
\end{equation}

Let $V(e)=e^\top Pe$, where $P\succ0$. Introduce a scalar $\nu\geq0$
such that
\begin{equation}
    P\preceq \nu W_0.
    \label{eq:initial-energy-bound}
\end{equation}
Then
\begin{equation}
    V(e_0)\leq\nu
    \qquad
    \text{for every }e_0\in\mathcal E_0.
\end{equation}

\subsubsection{Dissipation condition}

Introduce the change of variables
\begin{equation}
    Y:=PL,
    \qquad
    L=P^{-1}Y.
    \label{eq:observer-variable-substitution}
\end{equation}
For a prescribed decay rate $a>0$, define
\begin{align}
    \Phi(P,Y;a)
    &:=A^\top P+PA-C^\top Y^\top-YC+aP,
    \label{eq:design-Phi}\\
    \mathcal G(P,Y)
    &:=\sqrt{3}
    \begin{bmatrix}
        \bar w_TP & -\bar v_T(\delta)Y & \bar q_hYC
    \end{bmatrix}.
    \label{eq:design-G-PY}
\end{align}

If there exists $b\geq0$ such that
\begin{equation}
\begin{bmatrix}
    \Phi(P,Y;a) & \mathcal G(P,Y)\\
    \mathcal G(P,Y)^\top & -bI
\end{bmatrix}
\preceq0,
\label{eq:tube-design-lmi}
\end{equation}
then
\begin{equation}
    \dot V(t)
    \leq
    -aV(t)+b\|z(t)\|^2
    \leq
    -aV(t)+b.
    \label{eq:design-dissipation}
\end{equation}
Consequently,
\begin{align}
    V(e(t))&\leq \rho(t),\\
    \rho(t)&:=e^{-at}\nu
    +\frac{b}{a}\bigl(1-e^{-at}\bigr),\\
    &\hspace{10mm}t\in[0,T].
    \label{eq:design-rho}
\end{align}

The maximum Lyapunov level over the design horizon is
\begin{align}
    \rho_T
    :={}&\max_{t\in[0,T]}\rho(t)\\
    ={}&\max\left\{\nu,\,
      e^{-aT}\nu+\frac{b}{a}\bigl(1-e^{-aT}\bigr)\right\}.
    \label{eq:design-rho-T}
\end{align}

The corresponding componentwise tube radii satisfy
\begin{equation}
    |e_i(t)|
    \leq
    \sqrt{
        (P^{-1})_{ii}\rho(t)
    },
    \qquad
    i=1,\ldots,n.
    \label{eq:design-component-radius}
\end{equation}

\subsubsection{Normalized tube-design problem}

The product
\begin{equation}
    (P^{-1})_{ii}\rho_T
\end{equation}
is invariant under the common scaling
\begin{equation}
    (P,Y,b,\nu)
    \mapsto
    c(P,Y,b,\nu),
    \qquad c>0.
\end{equation}
The design problem can therefore be normalized by requiring
\begin{equation}
    \rho_T\leq1.
    \label{eq:rho-normalization}
\end{equation}
Using \eqref{eq:design-rho-T}, this condition is equivalent to
\begin{align}
    \nu
    &\leq1,
    \label{eq:rho-normalization-1}\\
    e^{-aT}\nu
    +
    \frac{1-e^{-aT}}{a}b
    &\leq1.
    \label{eq:rho-normalization-2}
\end{align}

Let $s_i>0$ denote the scale associated with the $i$th state component.
Introduce a scalar $\eta\geq0$ and impose
\begin{equation}
\begin{bmatrix}
    P & e_i\\
    e_i^\top & \eta s_i^2
\end{bmatrix}
\succeq0,
\qquad
i=1,\ldots,n,
\label{eq:component-radius-lmi}
\end{equation}
where $e_i$ is the $i$th canonical basis vector. By the Schur complement,
\begin{equation}
    (P^{-1})_{ii}
    \leq
    \eta s_i^2.
\end{equation}
Under the normalization $\rho_T\leq1$, this gives
\begin{equation}
    \sup_{t\in[0,T]}
    \frac{|e_i(t)|}{s_i}
    \leq
    \sqrt{\eta}.
\end{equation}

For a fixed decay rate $a>0$, the tube-oriented observer-design problem is
therefore
\begin{equation}
\begin{aligned}
    \min_{P,Y,b,\nu,\eta}\quad&\eta\\
    \text{s.t.}\quad
    &P\succ0,\quad b,\nu,\eta\geq0,\\
    &P\preceq\nu W_0,\quad \nu\leq1,\\
    &e^{-aT}\nu+\frac{1-e^{-aT}}{a}b\leq1,\\
    &\begin{bmatrix}
      \Phi(P,Y;a)&\mathcal G(P,Y)\\
      \mathcal G(P,Y)^\top&-bI
      \end{bmatrix}\preceq0,\\
    &\begin{bmatrix}P&e_i\\e_i^\top&\eta s_i^2\end{bmatrix}
      \succeq0,\quad i=1,\ldots,n.
\end{aligned}
\label{eq:tube-oriented-design-problem}
\end{equation}

For fixed $a$, all constraints in
\eqref{eq:tube-oriented-design-problem} are linear matrix inequalities in
the decision variables.

The observer gain corresponding to a feasible solution is recovered as
\begin{equation}
    L=P^{-1}Y.
    \label{eq:gain-recovery}
\end{equation}

\begin{algorithm}[t]
\caption{Finite-horizon high-probability tube-oriented observer design}
\label{alg:tube-oriented-observer-design}
\begin{algorithmic}[1]
\small
\Require $A,C,T,\delta,\bar h,\bar w_T,\bar v_T(\delta),M_x,W_0$, scales $s_i$, and candidate rates $\mathcal A$.
\State Set $\bar q_h\gets M_x\bar h$ and initialize the best objective value $\eta^\star\gets+\infty$.
\For{each $a\in\mathcal A$}
    \State Solve the fixed-$a$ LMI problem~\eqref{eq:tube-oriented-design-problem}.
    \If{the problem is feasible}
        \State Recover $L_a\gets P_a^{-1}Y_a$ and compute $\rho_a(t)$ from~\eqref{eq:design-rho}.
        \State Compute $r_{i,a}(t)\gets\sqrt{(P_a^{-1})_{ii}\rho_a(t)}$, $i=1,\ldots,n$.
        \If{$\eta_a<\eta^\star$}
            \State Set $a^\star\gets a$ and $(P^\star,Y^\star,L^\star)\gets(P_a,Y_a,L_a)$.
        \State Set $(b^\star,\nu^\star,\eta^\star)\gets(b_a,\nu_a,\eta_a)$.
        \EndIf
    \EndIf
\EndFor
\State Verify the selected solution in the original error dynamics and return $L^\star$, $P^\star$, and the componentwise tubes $r_i^\star(t)$.
\end{algorithmic}
\end{algorithm}

\medskip\noindent\textbf{Design verification.}\ Every accepted solution is reconstructed in the original variables $L=P^{-1}Y$. Positive definiteness, gain recovery, the maximum eigenvalue of each reconstructed LMI, and the componentwise radii are checked independently of the solver representation. The numerical values obtained from this design problem are reported later in the linear benchmark, after the model, disturbance envelopes, and component scales have been specified.

\subsection{Structured nonlinear sampled-data tube co-design}
\label{sec:structured-tube-design}

The nonlinear design is stated more compactly because it reuses the same design logic: specify channel envelopes, impose a quadratic dissipation inequality, and minimize the normalized componentwise tube radius. The additional element is the structured contraction condition. The known pair $(B_\phi,H_\phi)$ is retained in the LMI, which avoids a full-state Lipschitz overbound.

Consider \eqref{eq:structured-plant} with sampled measurements $y_k=Cx(t_k)+v_k$ and a zero-order-held output. Define the measured-output staleness
\begin{equation}
\zeta_h(t)=Cx(t)-Cx(\sigma(t)).
\end{equation}
The nonlinear sampled-data error dynamics are
\begin{equation}
\dot e=(A-LC)e+B_\phi\Delta\phi+w-L\bar v+L\zeta_h.
\label{eq:structured-sampled-error}
\end{equation}
Assume finite-horizon bounds
\begin{align}
\|w\|_{\infty,[0,T]}&\leq\bar w_T,\\
\|\bar v\|_{\infty,[0,T]}&\leq\bar v_T,\\
\|\zeta_h\|_{\infty,[0,T]}&\leq\bar\zeta_h.
\label{eq:structured-envelopes}
\end{align}
on the event under consideration. In the deterministic R1 study this event is the full bounded-disturbance set. In a high-probability application, the same construction applies on the corresponding horizon-level envelope event. These bounds are regional quantities: they are appropriate only on the operating set for which the disturbance and output-rate envelopes have been specified. The main restriction of the nonlinear design is therefore the quality of the selected operating region and its output-rate bound.

Use $Y=PL$ and introduce nonnegative channel coefficients $\mu_w$, $\mu_v$, and $\mu_h$. For a prescribed decay rate $a>0$, define the block matrix
\begin{equation}
\mathcal M_a=
\begin{bmatrix}
\Xi_a & PB_\phi & P & -Y & Y\\
B_\phi^\top P & -\tau I & 0&0&0\\
P&0&-\mu_w I&0&0\\
-Y^\top&0&0&-\mu_v I&0\\
Y^\top&0&0&0&-\mu_h I
\end{bmatrix},
\label{eq:structured-codesign-block}
\end{equation}
where
\begin{equation}
\Xi_a=\operatorname{He}(PA-YC)+aP
+\tau\gamma_\phi^2H_\phi^\top H_\phi.
\end{equation}
If $\mathcal M_a\preceq0$, then the S-procedure and the Schur complement imply
\begin{align}
\dot V\leq{}&-aV+\mu_w\|w\|^2+\mu_v\|\bar v\|^2\\
&+\mu_h\|\zeta_h\|^2.
\label{eq:structured-channel-dissipation}
\end{align}
Define
\begin{equation}
d_T^2=\mu_w\bar w_T^2+\mu_v\bar v_T^2+\mu_h\bar\zeta_h^2.
\label{eq:structured-dT}
\end{equation}
For the initial ellipsoid $\mathcal E_0=\{e_0:e_0^\top W_0e_0\leq1\}$, impose $P\preceq\nu W_0$. The resulting Lyapunov level satisfies
\begin{equation}
\rho(t)=e^{-at}\nu+\frac{1-e^{-at}}{a}d_T^2.
\label{eq:structured-design-rho}
\end{equation}

As in the linear case, a common positive scaling of $(P,Y,\tau,\mu_w,\mu_v,\mu_h,\nu)$ leaves the physical componentwise radii unchanged. For reference scales $s_i>0$, the fixed-$a$ co-design problem is
\begin{equation}
\begin{aligned}
\min_{P,Y,\tau,\mu_w,\mu_v,\mu_h,\nu,\eta}\quad&\eta\\
\text{s.t.}\quad
&P\succ0,\quad \tau,\mu_w,\mu_v,\mu_h,\nu,\eta\geq0,\\
&P\preceq\nu W_0,\quad \nu\leq1,\\
&e^{-aT}\nu+\frac{1-e^{-aT}}{a}d_T^2\leq1,\\
&\mathcal M_a\preceq0,\\
&\begin{bmatrix}P&e_i\\e_i^\top&\eta s_i^2\end{bmatrix}
\succeq0,\quad i=1,\ldots,n.
\end{aligned}
\label{eq:structured-tube-design-problem}
\end{equation}
For fixed $a$, all constraints are LMIs. The gain is recovered as $L=P^{-1}Y$, and $a$ is selected by an outer search over a specified grid. The objective bounds the worst normalized componentwise half-width by $\sqrt\eta$ under the chosen normalization.

The channel-specific formulation makes the design trade-off explicit. A larger gain can improve nominal contraction while increasing the measurement and staleness channels. The optimized design therefore balances contraction rate and gain size through the tube objective. Every reported solution is independently reconstructed using the recovered gain and the original block LMI.

\section{Numerical studies}
\label{sec:numerical}
\label{sec:numerical-studies}

The numerical study is primarily a study of tubes: their construction, their dependence on the disturbance envelope, and their use as an observer-design objective. The point estimators are included to provide familiar reference trajectories and pointwise uncertainty bands. The linear compartment model isolates the finite-horizon probability interpretation and the estimator comparison. The flexible-joint model tests whether the same tube-oriented design logic can be used with a structured nonlinear observer condition.

The numerical study uses two disturbance regimes. Regime R1 is the bounded-noise regime: process and measurement disturbances are componentwise bounded, and the resulting tube is a deterministic finite-horizon tube on the specified bounded-disturbance set. Regime R2 is the contaminated-measurement regime: the process disturbance remains bounded, while sampled measurement noise follows a Gaussian mixture with rare high-amplitude outliers. In R2 the tube is computed from a conservative high-probability measurement envelope over the full horizon. Thus R1 tests the deterministic bounded-disturbance construction, while R2 tests the finite-horizon probabilistic envelope under outlier-contaminated measurements. The nonlinear flexible-joint study uses the bounded R1 setting only, because its purpose is to test the structured nonlinear tube co-design rather than the contaminated-noise envelope.

For implementation, the procedure is deliberately sequential. First, select the observer structure, finite horizon, sampling period, state scales, and admissible initial-error set. Second, specify process, measurement, and held-output-mismatch envelopes, either deterministically or through a horizon-level probability event. Third, solve the gain-metric LMI for a selected decay-rate grid and recover the observer gain. Fourth, reconstruct the LMI residuals independently and compute the ellipsoidal or componentwise tube. Fifth, use the resulting time-varying tube as a monitoring object for alarms, constraint checks, or diagnostic escalation between sampled measurements.

\subsection{Common protocol, estimators, and metrics}
\label{sec:methodology}

The computational study has two roles. First, it verifies that the proposed LMIs produce feasible, independently reconstructed observer designs. Second, it compares the semantics and empirical behavior of simultaneous ISS tubes with pointwise uncertainty bands. All estimators in a given experiment use the same plant trajectory, disturbance realization, and sampled measurements.

\subsubsection{Estimators and uncertainty descriptions}

The linear study uses three estimator-side references: the continuous-discrete Kalman filter, Gaussian MHE, and robust Huber MHE. The Kalman filter is used as the classical linear Gaussian reference \cite{kalman1960new}. Gaussian MHE follows the constrained moving-horizon formulation of Rao et al. \cite{rao2003constrained}. Huber MHE keeps the same arrival and process terms and replaces the quadratic measurement penalty by a Huber loss; it is included only in the contaminated R2 regime as a robust-loss comparator, in line with MHE studies under bounded uncertainty and outliers \cite{liu2013robust}, \cite{liu2024multipleoutliers}.

The MHE implementation uses a $40$-step, $2$~s maximum window with exact discretization at the measurement period. Its decision vector contains the state sequence, the arrival covariance is $P_0=10^{-2}I_3$, the process covariance per sampling interval is $Q_d=5\times10^{-6}I_3$, and the measurement covariance is $R=2.5\times10^{-5}I_2$. Nonnegative-state constraints are imposed at every point in the window. During startup, the available horizon grows from the first measurement until the maximum window length is reached. Once the window slides, the arrival mean is the previous optimized state corresponding to the new window start, and the shifted previous solution provides the warm start. The terminal optimized state is returned as the current estimate. The Huber implementation uses threshold $1.5$ on whitened measurement residuals and four iteratively reweighted least-squares steps. Both MHE variants use OSQP with absolute and relative tolerances $10^{-5}$ and a limit of $2000$ iterations.

For both MHE variants, the local uncertainty band is obtained from the inverse Hessian at the MAP solution. With the implemented full-objective convention $J=\sum r^\top Wr$, the covariance approximation is $\Sigma=2H^{-1}$. This is a local Laplace approximation, with a different interpretation from a sampled posterior. When positivity constraints are active, the inverse Hessian omits the truncation induced by the active set, so active-constraint frequencies accompany the MHE results \cite{rao2003constrained}. The ISS observer is a continuous-time Luenberger-type observer driven by the sampled and held innovation; its reported uncertainty object is the finite-horizon tube derived in the preceding sections.

The nonlinear study compares the proposed nonlinear sampled-data ISS observer with a continuous-discrete EKF. The EKF uses RK4 propagation of the nonlinear mean, local Jacobian covariance propagation, and a Joseph-form covariance update at measurement instants. It is used as a standard local Gaussian reference for nonlinear estimation \cite{haseltine2005critical}.

\subsubsection{Metrics}

Pointwise inclusion is the empirical inclusion frequency over recorded state-time pairs. Componentwise simultaneous inclusion is the fraction of trajectories for which one state component remains in its band for every recorded time. Joint trajectory inclusion requires every state to remain inside its corresponding band for the complete horizon; for ISS tubes, the same event is termed simultaneous containment. The aggregate normalized RMSE is the root mean square of the statewise normalized RMSE values. The mean width is averaged over states, trajectories, and recorded times. Wilson score intervals are reported only for trajectory-level binary frequencies because state-time observations within one trajectory are dependent.

Low joint inclusion of pointwise $95\%$ bands over a long grid is compatible with well-calibrated marginal pointwise events. A product of many pointwise events can have low probability even when each marginal event is well calibrated. This distinction is central to the comparison.

\subsubsection{Numerical verification of LMI solutions}

Solver status alone is insufficient evidence of feasibility. Each accepted design is reconstructed from $L=P^{-1}Y$, after which symmetry, positive definiteness, gain recovery, normalization, componentwise inequalities, and maximum LMI eigenvalues are evaluated independently. The paper reports scalar margins and scientifically interpretable objectives. The public Zenodo replication package contains the full matrices, residual arrays, configurations, and figure and table generators used to reproduce these checks \cite{baranowski2026replication}.

\subsection{Linear compartment benchmark}
\label{sec:linear-study}

\subsubsection{Model, measurements, and disturbance regimes}

The linear benchmark is a nominally positive compartment system
\begin{equation}
\begin{aligned}
\dot x(t)&=Ax(t)+w(t),\\
A&=\begin{bmatrix}
-0.65&0.15&0\\
0.35&-0.60&0.10\\
0&0.25&-0.35
\end{bmatrix}.
\end{aligned}
\end{equation}
with aggregated measurements
\begin{equation}
y_k=Cx(t_k)+v_k,
\qquad
C=\begin{bmatrix}1&1&1\\1&0&0\end{bmatrix}.
\end{equation}
The first output is total mass and the second is the first compartment. The physically meaningful state space is the positive cone $\mathbb R_{\geq0}^3$, a nonlinear constraint set. This matters for estimation because a linear observer can estimate a positive plant while producing negative components during transients or under noise. The benchmark is therefore nominally positive in the plant model, while positivity of the estimates is enforced only for the constrained MHE comparator. The ISS tube itself bounds signed estimation error around the ISS observer trajectory and does not impose positivity of the estimate.

The two regimes introduced at the beginning of this section are instantiated as follows. In R1, the process and measurement disturbances are componentwise bounded by $\bar w=0.02$ and $\bar v=0.01$. In R2, the process disturbance remains bounded and the sampled measurement noise follows a contaminated Gaussian mixture: standard deviation $\sigma=0.005$, contamination probability $p=0.03$, and outlier scale multiplier $\kappa=15$. The reported study uses $T=20$~s, $\Delta t_y=0.05$~s, $\delta=0.05$, and $N=1000$ trajectories per regime. Initial states are sampled independently from $[0.5,1.5]^3$, and all estimators start from the true initial state. Process disturbances are sampled independently and uniformly from $[-0.02,0.02]^3$ and held for $0.2$~s. R1 measurement errors are sampled independently and uniformly from $[-0.01,0.01]^2$ at each measurement time; R2 measurement errors follow the mixture specified above. All estimators use the same realization within a trajectory. The perfect initialization is a simulation protocol choice; the tube design still uses the radius-$0.1$ initial-error set.

For R2, let $K=\lfloor T/\Delta t_y\rfloor=400$ and $m=2$. Because the process disturbance is deterministically bounded, the complete risk budget $\delta$ is assigned to the measurement channel. The componentwise mixture envelope is constructed from the worst mixture scale and a two-sided union bound over two mixture components, $m$ channels, and $K+1$ grid instants,
\begin{equation}
\bar v_{T,\mathrm{comp}}(\delta)
=\kappa\sigma\sqrt{2\log\!\left(\frac{4m(K+1)}{\delta}\right)}
=0.3529.
\end{equation}
Thus $\bar v_T=\sqrt m\,\bar v_{T,\mathrm{comp}}=0.4991$. For the selected ISS observer, the process-plus-held-noise contribution is
\begin{equation}
\sqrt{3}\,\bar w+\lVert L\rVert_2\bar v_T=0.9670,
\end{equation}
before the separately bounded held-output mismatch channel is added. The construction uses a conservative sub-Gaussian mixture-scale bound; for a pure Gaussian model an exact Gaussian quantile would be narrower. The simulation has $400$ measurement updates and $401$ state-grid points; using $K+1=401$ potential noise instants in the union bound also covers the terminal grid point and is therefore marginally conservative.

\subsubsection{Representative finite-horizon tubes}

Figure~\ref{fig:linear-overview} shows representative componentwise ISS tubes over the full $T=20$~s horizon. The left column uses the bounded R1 disturbance model, whereas the right column uses the contaminated R2 measurement model and its finite-horizon high-probability envelope. The widening in R2 reflects the stronger horizon-level envelope required to accommodate rare, large measurement errors; it should be read as a simultaneous robustness envelope with a stronger horizon-level interpretation than a pointwise Gaussian uncertainty band.

\begin{figure*}[pos=htbp]
\centering
\includegraphics[width=0.92\textwidth]{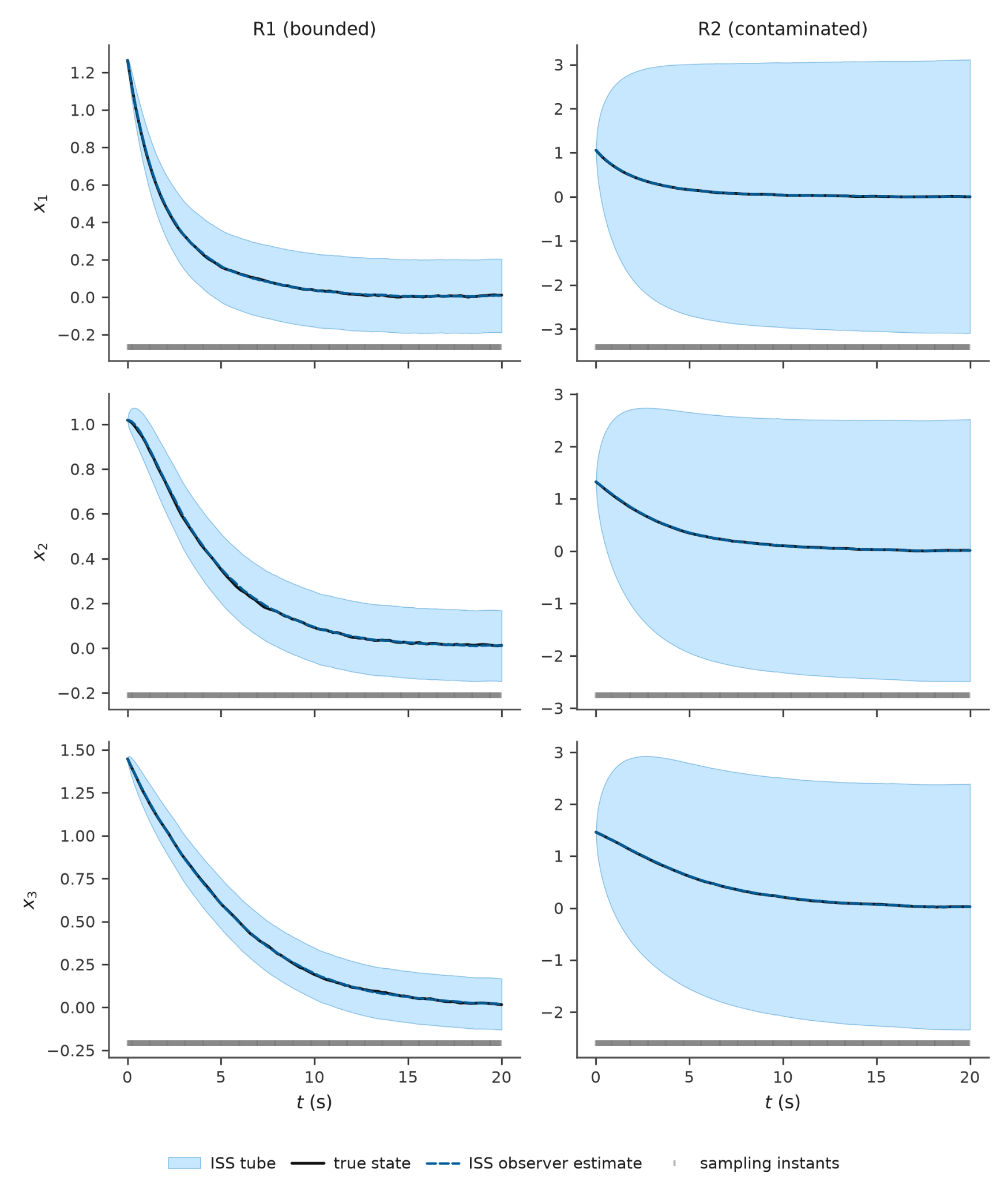}
\caption{Representative componentwise ISS tubes for the linear compartment model. Rows show the three states and columns show R1 and R2. Solid curves are true states, dashed curves are observer estimates, shading denotes the ISS tubes, and gray marks indicate sampling instants. The wider R2 tube results from the finite-horizon envelope for contaminated measurement noise.}
\label{fig:linear-overview}
\end{figure*}

\subsubsection{Tube-oriented linear design}

The co-design uses an initial-error ball of radius $r_0=0.1$ and the bounded R1 envelopes. Figure~\ref{fig:linear-codesign} shows a nonmonotone dependence of the width objective on the prescribed decay rate. The optimum occurs at $a=0.5$. The optimized normalized half-widths are approximately $(0.1370,0.1538,0.1538)$, whereas the baseline produces $(0.2226,0.1574,0.2174)$. The improvement is therefore a redistribution as well as a reduction of the worst componentwise bound.

The envelope decomposition is also informative. The staleness contribution $0.0396$ is slightly larger than the process contribution $0.0346$, while the direct measurement contribution $0.0038$ is much smaller. This result supports treating sampling staleness as an explicit design channel.

\begin{figure*}[pos=htbp]
\centering
\includegraphics[width=0.94\textwidth]{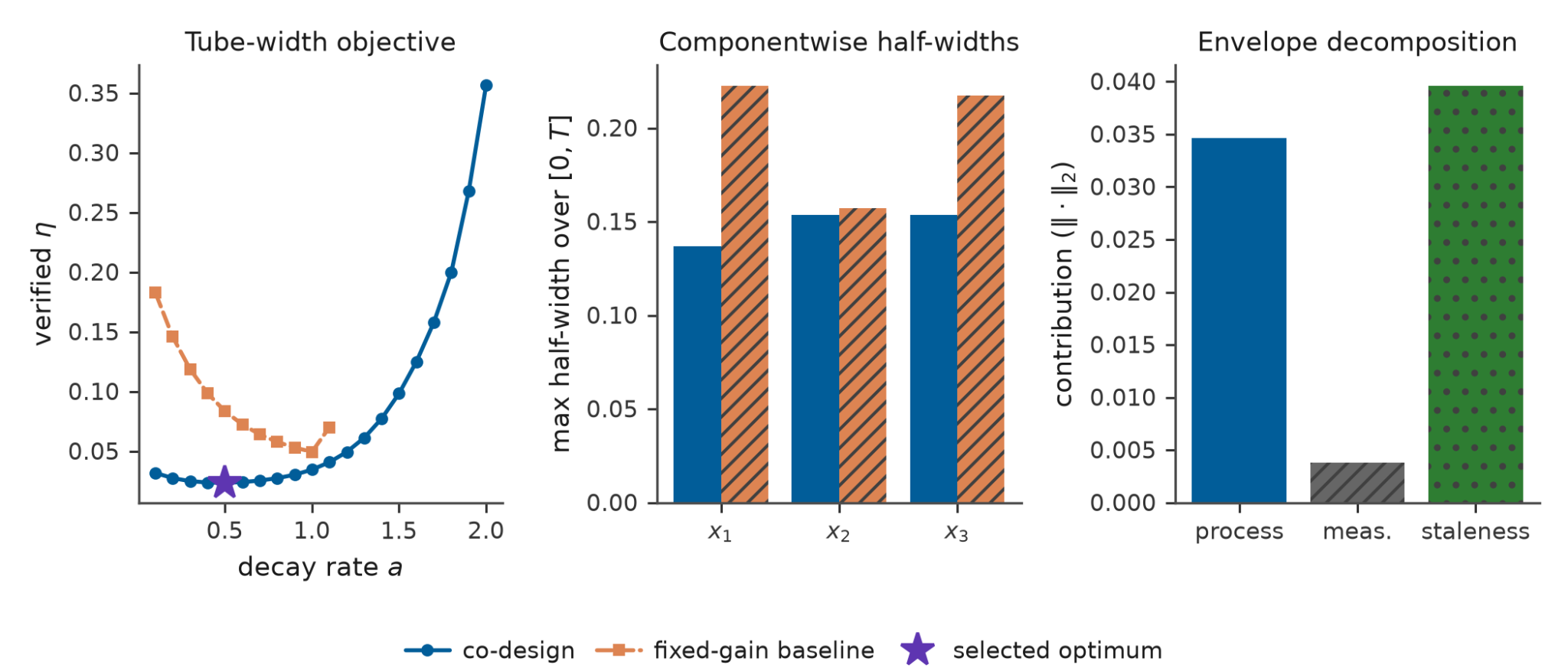}
\caption{Linear tube-oriented observer co-design under the bounded R1 envelopes. The left panel shows the verified tube-width objective over the prescribed decay-rate search for the co-designed and fixed-gain observers; the star marks the selected optimum at $a=0.5$. The middle panel compares the resulting maximum componentwise half-widths over $[0,T]$. The right panel decomposes the disturbance envelope into process, direct-measurement, and sample-and-hold staleness contributions, with staleness slightly exceeding the process contribution.}
\label{fig:linear-codesign}
\end{figure*}

\subsubsection{Estimator comparison over 1000 trajectories}

Table~\ref{tab:linear-comparison-metrics} summarizes the $N=1000$, $T=20$ comparison. In R1, the Kalman filter has the lowest aggregate normalized RMSE, $0.00338$, followed by the ISS observer at $0.00385$ and Gaussian MHE at $0.00712$. In R2, the ISS observer has the lowest aggregate normalized RMSE, $0.00422$, while the Kalman filter, Huber MHE, and Gaussian MHE attain $0.00565$, $0.00890$, and $0.01393$, respectively. Relative to Gaussian MHE in R2, Huber MHE reduces aggregate normalized RMSE by $36.1\%$ and increases mean pointwise inclusion from $0.917$ to $0.968$, with essentially unchanged mean band width.

None of the pointwise Kalman or MHE bands jointly includes all states at all recorded times over the complete horizon. The ISS tube contains all $1000$ trajectories in both regimes; the corresponding $95\%$ Wilson interval for the unknown trajectory-containment probability is $[0.9962,1]$. The stronger simultaneous statement is accompanied by much greater width. In R1, the mean ISS width is $0.319$, compared with $0.0230$ for the Kalman band and $0.0471$ for Gaussian MHE. In R2, the finite-horizon contaminated-noise envelope increases the mean ISS width to $4.853$, while the pointwise Kalman and MHE widths remain approximately $0.0230$ and $0.0471$. This contrast quantifies the conservatism introduced by the horizon-level envelope.

\begin{table*}[pos=htbp]
\centering
\caption{Linear estimator and uncertainty metrics for $N=1000$ trajectories over $T=20$~s. Aggregate normalized RMSE is the root mean square of the three statewise normalized RMSE values; the linear study uses unit state scales. Mean pointwise inclusion and mean width are averaged over the three state components. ``Joint trajectory inclusion'' is the fraction of trajectories for which all state components remain inside the reported band or tube at every recorded time. For the ISS observer, inclusion denotes simultaneous tube containment under the corresponding disturbance-envelope assumption.}
\label{tab:linear-comparison-metrics}
\small
\setlength{\tabcolsep}{5pt}
\begin{tabular}{@{}llrrrr@{}}
\toprule
Regime & Estimator & agg. nRMSE & mean ptwise incl. & joint traj. incl. & mean width \\
\midrule
R1 & continuous-discrete KF & 0.00338 & 0.991 & 0.000 & 0.0230 \\
R1 & Gaussian MHE           & 0.00712 & 0.972 & 0.000 & 0.0471 \\
R1 & ISS observer           & 0.00385 & 1.000 & 1.000 & 0.319 \\
\addlinespace
R2 & continuous-discrete KF & 0.00565 & 0.950 & 0.000 & 0.0230 \\
R2 & Gaussian MHE           & 0.01393 & 0.917 & 0.000 & 0.0471 \\
R2 & robust Huber MHE       & 0.00890 & 0.968 & 0.000 & 0.0472 \\
R2 & ISS observer           & 0.00422 & 1.000 & 1.000 & 4.853 \\
\bottomrule
\end{tabular}

\end{table*}

The positivity constraints are active in approximately $15.8\%$--$30.7\%$ of the MHE state-update records, depending on state and regime, so the inverse-Hessian bands must be interpreted as local curvature approximations affected by the active constraints. OSQP returns a non-success status for $3713/400000$ R1 Gaussian updates ($0.928\%$), $3300/400000$ R2 Gaussian updates ($0.825\%$), and $61/400000$ R2 Huber updates ($0.0153\%$). In every such case the best available iterate is retained and included in all metrics. In serial execution on an Apple M2 Ultra with numerical-library threads restricted to one, the corrected MHE computation requires approximately $0.626$~s per R1 Gaussian trajectory, $0.630$~s per R2 Gaussian trajectory, and $2.001$~s per R2 Huber trajectory. These times refer to the MHE computation alone.

\subsection{Nonlinear flexible-joint benchmark}
\label{sec:nonlinear-study}

\subsubsection{Structured benchmark and regional staleness model}

The nonlinear benchmark follows the single-link flexible-joint example associated with Lipschitz observer design \cite{rajamani1998observers}. The state ordering is
\begin{equation}
x=[\theta_m,\ \omega_m,\ \theta_\ell,\ \omega_\ell]^\top,
\end{equation}
and the model is represented in the structured form \eqref{eq:structured-plant} with
\begin{align}
B_\phi&=\begin{bmatrix}0&0&0&-3.33\end{bmatrix}^\top,\\
H_\phi&=\begin{bmatrix}0&0&1&0\end{bmatrix},\\
\phi(s)&=\sin s.
\end{align}
Both angles are measured,
\begin{equation}
C=\begin{bmatrix}
1&0&0&0\\
0&0&1&0
\end{bmatrix},
\end{equation}
and the known excitation is $u(t)=0.5\sin(2t)$. The sampling period is $\Delta t_y=0.01$~s. Process and measurement envelope norms are $\bar w_T=0.04$ and $\bar v_T=0.0141$.

The output-staleness envelope is regional. The design assumes that the measured-channel velocity remains in the regional set
\begin{equation}
    |\dot\theta_m(t)|\leq \Omega_m,
    \qquad
    |\dot\theta_\ell(t)|\leq \Omega_\ell,
    \qquad t\in[0,T],
\end{equation}
with $\Omega=(6.758,4.750)$~rad/s. Under this regional assumption,
\begin{align}
M_y&=\sqrt{\Omega_m^2+\Omega_\ell^2}=8.260,\\
\bar\zeta_h&=M_y\Delta t_y=0.0826.
\end{align}
These quantities were calibrated from a 30-run preliminary study with a $1.2$ safety factor. They are therefore calibrated regional quantities, not global analytical bounds. The theorem-level statement is conditional on the trajectory remaining in this measured-channel velocity region.

The initial-error set is the Euclidean ball $\|e_0\|\leq0.5$. Reference scales are $s=(1,8,1,8)$, with angle scales in radians and velocity scales in radians per second. Sensitivity artifacts archived in the replication package show that the objective is effectively flat once the velocity reference scale reaches $8$ \cite{baranowski2026replication}.

\subsubsection{Structured contraction condition}

The structured LMI is feasible for the two-angle output configuration and yields the independently verified contraction rate $\alpha=3.62$. The reconstructed maximum LMI eigenvalue is $-1.4\times10^{-4}$ and $\operatorname{cond}(P)=15.08$. The tested motor-angle-only and full-norm constant-metric sufficient conditions were infeasible. This feasibility outcome is limited to those sufficient conditions and leaves open other observers, metrics, and local conditions.

\begin{figure*}[pos=htbp]
\centering
\includegraphics[width=0.94\textwidth]{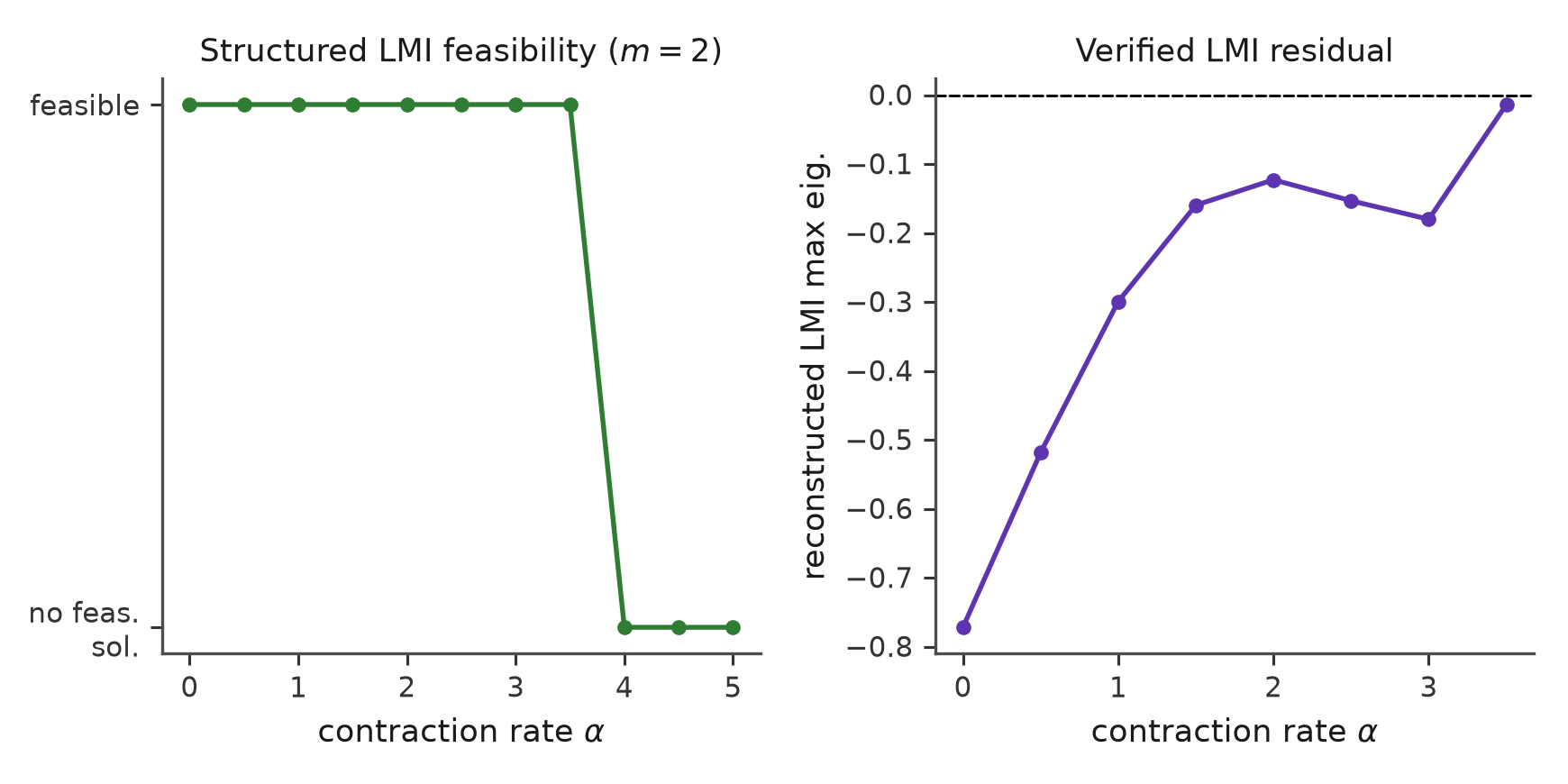}
\caption{Structured contraction study for the two-output configuration. The left panel shows feasibility over the rate sweep; the right panel shows the maximum eigenvalue of the reconstructed LMI. Negative values satisfy the tested inequality, and the dashed line marks the feasibility boundary.}
\label{fig:nonlinear-contraction}
\end{figure*}

\subsubsection{Nonlinear tube co-design}

The conventional contraction-oriented baseline uses a minimum-gain design with target contraction rate $2$, $\|L\|=64.0$, and worst normalized half-width $J_{\rm old}=11.31$. The tube-oriented design searches over the decay rate and selects $a=1.333$, with $\eta_{\rm verified}=0.2556$, $\|L\|=90.1$, and reconstructed block-LMI maximum eigenvalue $-3.0\times10^{-7}$. The resulting worst normalized half-width is $J_{\rm new}=0.506$.

The ratio
\begin{equation}
R_J=\frac{J_{\rm new}}{J_{\rm old}}=0.0447
\end{equation}
corresponds to an approximately $22.4$-fold reduction in the worst normalized objective. A rate-matched calculation at the baseline dissipation rate $a=1.0$ gives a co-designed objective $0.5057$ and an improvement factor $22.37$, confirming that the reported reduction is due to tube-oriented co-design and not to the selected decay rate. This factor concerns the worst normalized objective. The angle bounds show the largest relative improvement, while the normalized velocity bounds reflect their larger reference scales.

\begin{figure*}[pos=htbp]
\centering
\includegraphics[width=0.94\textwidth]{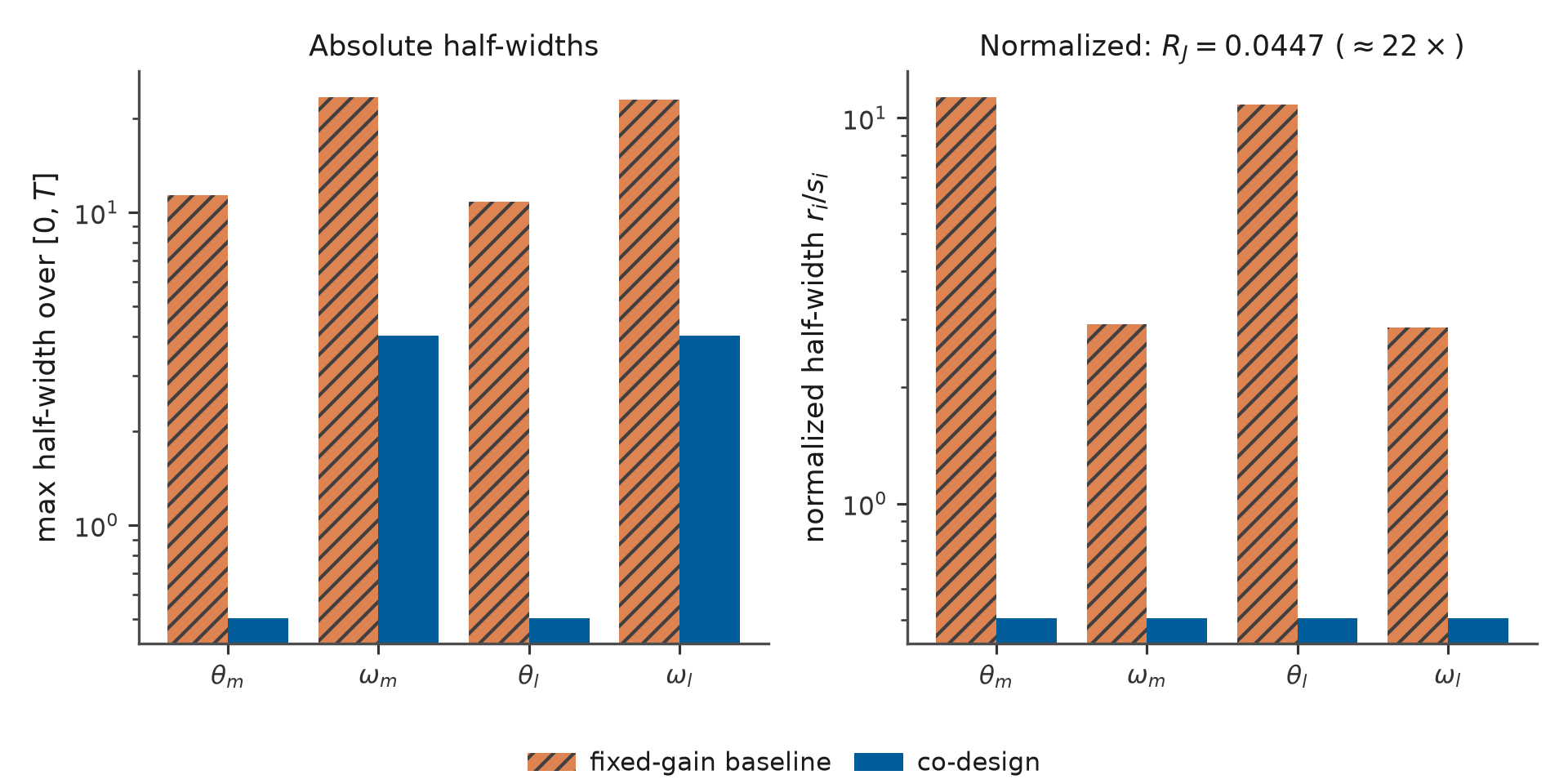}
\caption{Nonlinear flexible-joint tube-design comparison. The figure reports maximum componentwise half-widths for the contraction-oriented baseline and the tube co-design in the nonlinear benchmark. The left panel shows absolute values; the right panel shows values normalized by the specified state scales. The ratio $R_J=0.0447$ compares the worst normalized half-widths.}
\label{fig:nonlinear-improvement}
\end{figure*}

\subsubsection{Containment, utilization, and estimator behavior}

The bounded R1 Monte Carlo validation uses $N=\NonlinearValidationN$ trajectories. Joint simultaneous containment is $1.0$, no operating-region violation is observed, and the largest output-rate ratio is $\max_t\|\dot y(t)\|_2/M_y=0.706$. A separate deterministic stress test evaluates $195$ boundary-oriented initial-error and extremal disturbance cases. Its largest output-rate ratio is $0.605$, its largest held-output-mismatch ratio is $0.716$, and it produces no region or tube violation. The random study is more adverse in output rate than this structured test, so the latter is supporting evidence rather than a worst-case proof. Neither result converts the calibrated $M_y$ into a global analytical bound.

Table~\ref{tab:nonlinear-summary} collects the principal nonlinear design and validation quantities. The disturbance-channel coefficients show that the staleness channel has the largest coefficient in the optimized dissipation inequality, consistent with the sampling-interval sensitivity artifacts archived in the replication package \cite{baranowski2026replication}.

\begin{table*}[pos=htbp]
\centering
\caption{Selected nonlinear tube-design and bounded-disturbance validation metrics. The independently verified refined contraction rate $\alpha=3.62$ and the corresponding LMI residual sweep are reported in Figure~\ref{fig:nonlinear-contraction}.}
\label{tab:nonlinear-summary}
\small
\setlength{\tabcolsep}{4pt}
\begin{tabular}{@{}llr@{}}
\toprule
Category & Quantity & Value \\
\midrule
Tube design & decay rate $a$ & $1.333$ \\
            & verified $\eta$ & $0.2556$ \\
            & $\lVert L\rVert_2$ & $90.11$ \\
            & baseline objective $J_{\rm old}$ & $11.313$ \\
            & co-designed objective $J_{\rm new}$ & $0.5056$ \\
            & free-rate improvement factor & $22.38$ \\
            & fixed-rate factor at $a=1$ & $22.37$ \\
\addlinespace
Validation  & random trajectories & $500$ \\
            & joint simultaneous containment & $1.000$ \\
            & structured stress cases & $195$ \\
            & region / tube violations & $0/0$ \\
            & max. random output-rate ratio & $0.706$ \\
            & max. stress output-rate ratio & $0.605$ \\
            & maximum normalized utilization & $0.986$ \\
\addlinespace
Regional model & output-rate bound $M_y$ & $8.260$ \\
               & process coefficient $\mu_w$ & $0.0117$ \\
               & measurement coefficient $\mu_v$ & $0.193$ \\
               & staleness coefficient $\mu_h$ & $1.128$ \\
\bottomrule
\end{tabular}

\end{table*}

The maximum normalized utilization is $0.986$ and occurs at $t=0$ in $\theta_m$, where the sampled initial error nearly saturates the prescribed radius-$0.5$ initial set. One decay time for the selected design is $1/a=0.75$~s; after this time the maximum utilization falls to $0.213$. Thus co-design makes the tube usable, while the maximum-utilization result primarily reflects initial-set dominance; post-transient tightness should be assessed separately.

\begin{figure*}[pos=htbp]
\centering
\includegraphics[width=0.92\textwidth]{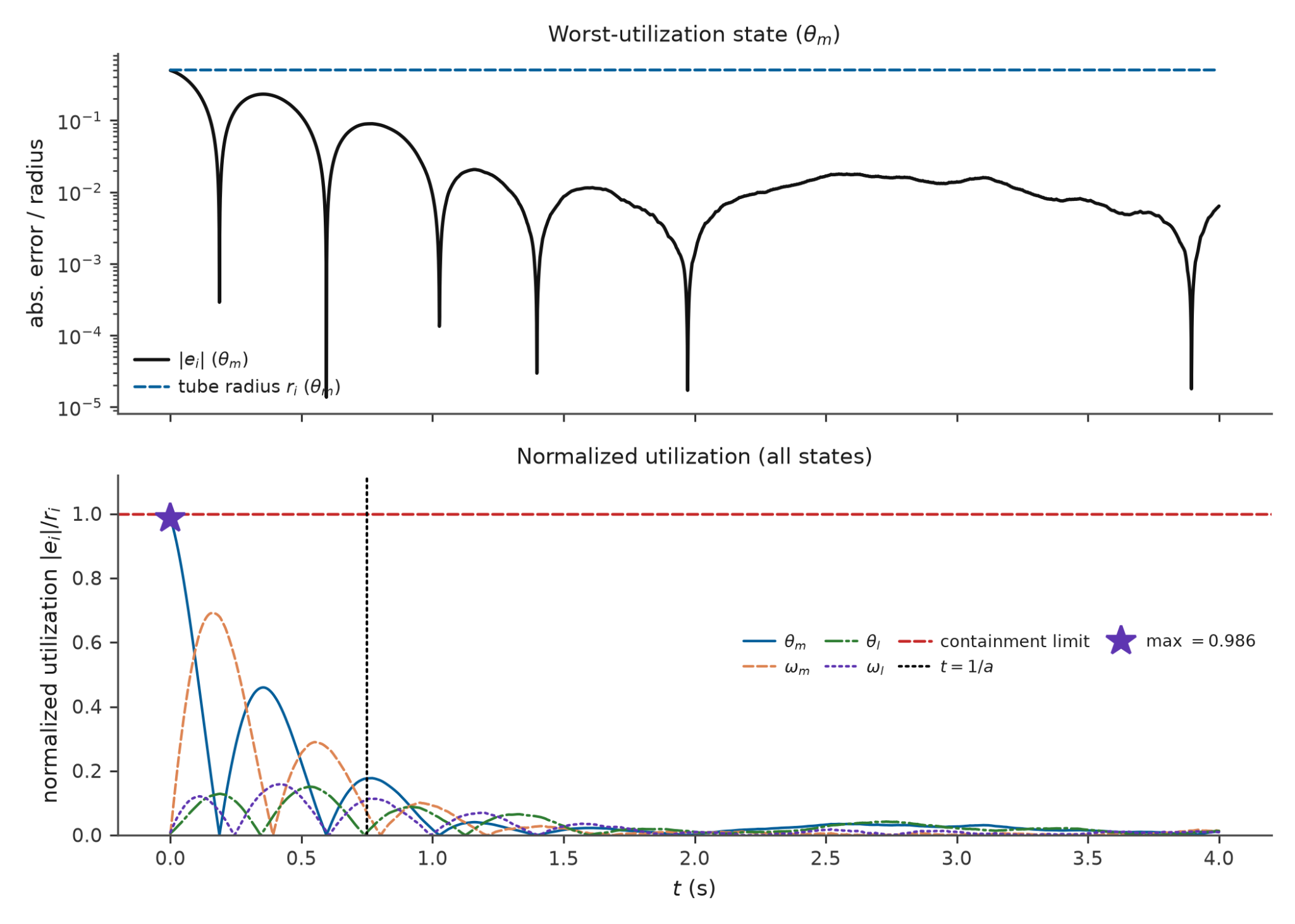}
\caption{Nonlinear ISS-tube utilization for the worst Monte Carlo realization. The upper panel compares $|e_{\theta_m}|$ with its radius; the lower panel shows $|e_i|/r_i$ for all states. The star marks the maximum $0.986$ at $t=0$, and the vertical line marks one decay time, $1/a=0.75$~s.}
\label{fig:nonlinear-utilization}
\end{figure*}

\begin{figure*}[pos=htbp]
\centering
\includegraphics[width=0.92\textwidth]{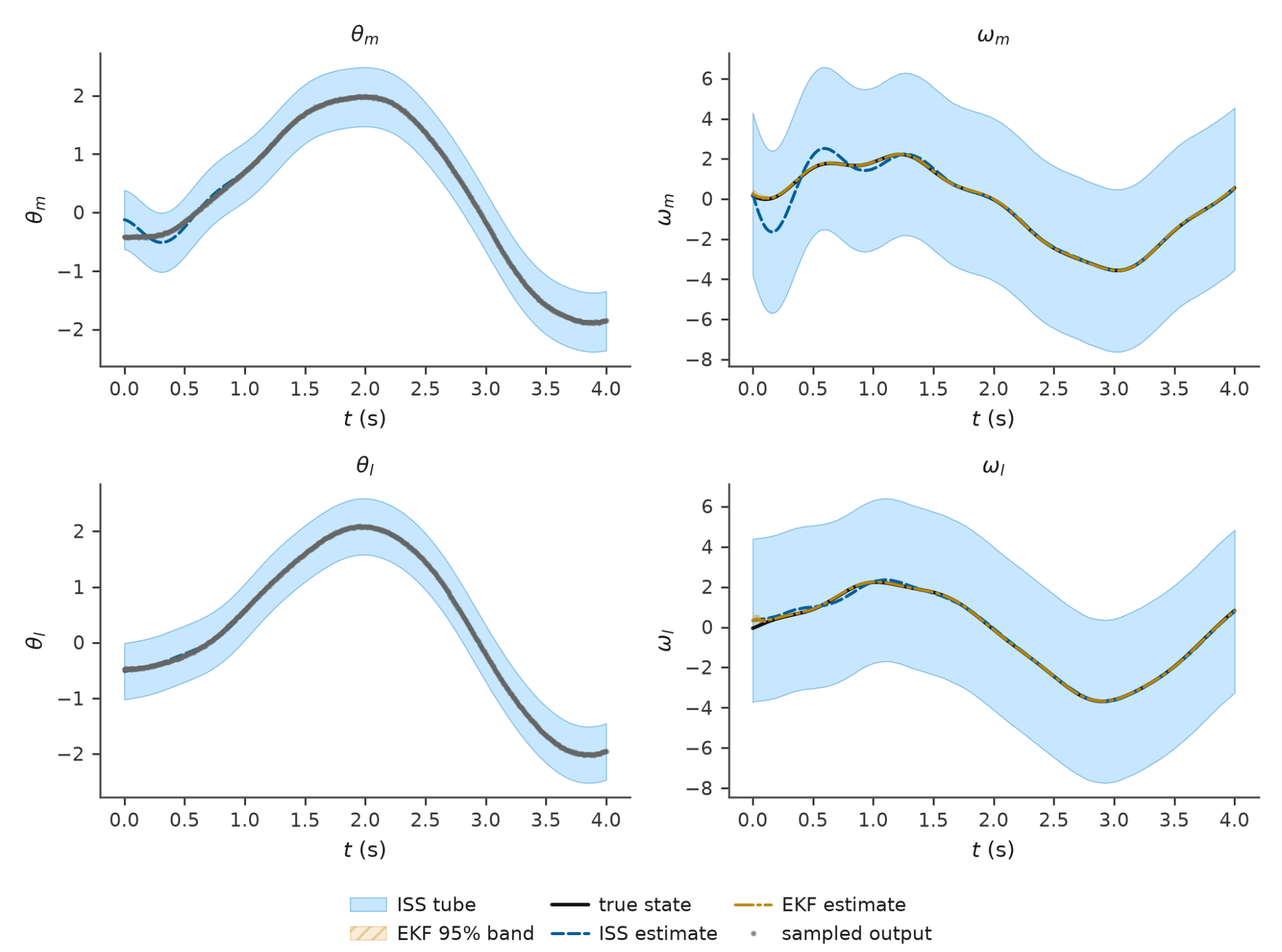}
\caption{Representative nonlinear estimation trajectory for the flexible-joint benchmark. The four panels show the motor angle $\theta_m$, motor velocity $\omega_m$, link angle $\theta_\ell$, and link velocity $\omega_\ell$. Solid black curves denote the true states, dashed blue curves the sampled-data ISS-observer estimates, and dash-dotted curves the EKF estimates. The blue shaded regions are the componentwise ISS tubes, while the hatched regions are the EKF pointwise $95\%$ bands. Sampled outputs are shown for the two measured angles.}
\label{fig:nonlinear-estimator-comparison}
\end{figure*}

\subsection{Cross-study interpretation}
\label{subsec:cross-study}

The two numerical studies evaluate complementary parts of the same pipeline. The linear compartment example isolates the probability semantics: the KF and MHE bands are narrow pointwise uncertainty descriptions, whereas the ISS construction gives a simultaneous finite-horizon tube. This is why the ISS tube has perfect trajectory inclusion in the reported Monte Carlo study while being substantially wider, especially under the contaminated R2 envelope. The same example also shows that optimizing the observer for tube width changes the engineering quantity being minimized together with the nominal convergence rate.

The flexible-joint example tests whether the same tube-oriented design logic survives a nonlinear observer condition. The structured LMI uses the known location of the sine nonlinearity while avoiding a full-state Lipschitz overbound. This makes the constant-metric co-design feasible and produces a much smaller normalized tube. The linear study therefore demonstrates the finite-horizon uncertainty interpretation and estimator comparison, while the nonlinear study demonstrates that the design principle extends beyond the linear compartment model. Together they support the central claim: sampled-data estimation should account explicitly for intersample mismatch when the desired uncertainty statement is simultaneous over a finite horizon.

\section{Discussion}
\label{sec:discussion}

\subsection{Interpretation of the numerical studies}

The two numerical studies evaluate different parts of the same sampled-data tube construction. The linear compartment example isolates the finite-horizon probability semantics, the influence of the contaminated-noise envelope, and the comparison with Kalman and MHE uncertainty bands. The flexible-joint example tests the same design logic in a nonlinear setting where the contraction condition uses the known structure of the nonlinear channel. In both studies, the tube is evaluated as a simultaneous containment object over the complete horizon, while KF, EKF, and MHE bands are used as pointwise estimator-side comparators.

The results also show where conservatism enters. Under nominal Gaussian noise, pointwise filters and MHE bands are much narrower than the ISS tube. Under contaminated measurements, the robust Huber MHE improves pointwise behavior relative to Gaussian MHE, while the ISS tube remains a horizon-level robustness statement driven by the prescribed worst-scale envelope. The comparison therefore separates two tasks: local statistical uncertainty reporting and finite-horizon containment under explicit disturbance envelopes.

\subsection{Practical interpretation for sampled-data monitoring}

The held-output mismatch appears as an additional input channel in the sampled-data error dynamics. Its contribution depends on the observer gain, the sampling period, and the output-rate bound. This channel is easy to overlook when a continuous-time observer is written formally with sampled measurements, but it becomes explicit in the tube radius and in the co-design problem. Reducing the observer gain may attenuate sampled measurement noise and staleness, while increasing the gain may improve nominal contraction. The proposed metric and gain optimization balances these effects for the selected horizon, component scales, and disturbance envelope.

\subsection{Limitations and scope}

The guarantees are finite-horizon and conditional on the specified disturbance envelope. Structural model errors outside the aggregated disturbance channels are outside the stated assumptions. In R2, the worst-scale Gaussian-mixture tail bound and the union bound provide a transparent simultaneous event, with visible conservatism relative to pointwise stochastic intervals. Dependence-aware extreme-value bounds could reduce this conservatism while preserving the deterministic ISS implication.

The flexible-joint result is regional because the output-rate quantity $M_y$ and the resulting staleness bound are calibrated for a prescribed operating region. Zero region and tube violations in $500$ random trajectories and $195$ structured boundary-oriented cases support the numerical use of this assumption, but they do not turn the calibrated value into a global bound. The result should therefore be read as a regional sampled-data design with empirical stress testing, while an analytical regional velocity bound remains future work.

For constrained MHE, the inverse Hessian describes local curvature around the optimizer and is affected by active positivity constraints. The resulting Laplace bands are therefore local uncertainty diagnostics for comparison with the ISS tube. Experimental deployment, embedded timing assessment, state-dependent contraction metrics, and joint optimization of the observer and sampling policy remain outside the present study.

\section{Conclusions}
\label{sec:conclusions}

This paper develops uniform finite-horizon ISS tubes for continuous-time state estimation with sampled and held measurements. One horizon-level disturbance-envelope event yields simultaneous containment of the complete estimation-error trajectory. Quadratic dissipation inequalities provide ellipsoidal, Euclidean, and componentwise bounds, while the sampled-data error model separates process disturbance, held measurement noise, and intersample output mismatch. The resulting quantities can be evaluated online once the observer and metric have been designed offline.

The tube also serves as a practical design objective. In the linear compartment benchmark, joint gain-metric design reduces the worst normalized componentwise half-width from $0.2226$ to $0.1538$. In the nonlinear flexible-joint benchmark, the structured S-procedure condition is feasible for two measured angles and gives a verified contraction rate $\alpha=3.62$. The co-designed tube improves the normalized half-width by approximately a factor of $22.4$, with a rate-matched factor of $22.37$. All $500$ random trajectories and $195$ structured stress cases remain inside the nonlinear tube under the specified regional assumption.

The $N=1000$, $T=20$ linear study illustrates the operational distinction between pointwise uncertainty and complete-trajectory containment. The ISS tube contains all tested trajectories in both regimes, at the cost of greater width, especially under the conservative contaminated-noise envelope. Under R2 contamination, robust Huber MHE improves pointwise inclusion and aggregate normalized RMSE relative to Gaussian MHE. These results support the use of the ISS tube as a robustness-oriented monitoring bound alongside statistical estimator uncertainty.

The scope remains finite-horizon and model dependent, and the nonlinear guarantee is regional. Future work will focus on less conservative disturbance envelopes, analytical regional output-rate bounds, state-dependent contraction metrics, sampling-policy co-design, and experimental validation in embedded monitoring systems.

\section*{Acknowledgments}
The author acknowledges the scientific influence and guidance of Professor Wojciech Mitkowski.

\section*{Funding}
This work was financed by AGH University of Krakow through its subvention for scientific research.

\section*{Declaration of competing interest}
The author declares that he has no known competing financial interests or personal relationships that could have appeared to influence the work reported in this paper.

\section*{Data and code availability}
The numerical examples use model-generated data. The source code, configuration files, deterministic seeds, generated numerical artifacts, verification records, and scripts required to reproduce the reported figures and tables are available in the public Zenodo replication package \cite{baranowski2026replication}, DOI: \href{https://doi.org/10.5281/zenodo.21010764}{10.5281/zenodo.21010764}.

\section*{Declaration of generative AI and AI-assisted technologies in the writing process}
During preparation of the manuscript, the author used OpenAI ChatGPT for language refinement and structural review. The author reviewed and edited all scientific content, numerical results, interpretations, and references and takes full responsibility for the content of the publication.

\sloppy
\bibliographystyle{elsarticle-num-names}
\bibliography{references}

\clearpage
\appendix
\section*{Appendix material}
\addcontentsline{toc}{section}{Appendix material}
\setcounter{section}{0}
\renewcommand{\thesection}{A.\arabic{section}}
\setcounter{equation}{0}
\renewcommand{\theequation}{A.\arabic{section}.\arabic{equation}}
\section{Auxiliary finite-horizon observations}
\label{app:auxiliary}

For a quadratic tube level
\begin{equation}
\rho(t)=e^{-at}\rho_0+(1-e^{-at})\rho_\infty,
\end{equation}
the derivative has the sign of $\rho_\infty-\rho_0$. Hence the maximum over $[0,T]$ is attained at $0$ when $\rho_0\geq\rho_\infty$ and at $T$ otherwise. This elementary fact produces the two normalization inequalities used in both co-design problems.

The componentwise projection follows from a constrained quadratic maximization. For $P\succ0$,
\begin{equation}
\max_{e^\top Pe\leq\rho}e_i^2=\rho(P^{-1})_{ii}.
\end{equation}
The Schur-complement inequality
\begin{equation}
\begin{bmatrix}P&e_i\\e_i^\top&\eta s_i^2\end{bmatrix}\succeq0
\end{equation}
is therefore equivalent to $(P^{-1})_{ii}\leq\eta s_i^2$.

\section{Linear co-design transformations}
\label{app:linear-codesign}

The substitution $Y=PL$ removes the bilinear product $PL$ from the nominal observer dynamics. For fixed $a$, the process, measurement, and staleness blocks are affine in $(P,Y)$. The normalized finite-horizon problem remains convex because the factors $e^{-aT}$ and $(1-e^{-aT})/a$ are constants in the inner problem. The outer rate search is therefore one-dimensional and does not alter the feasibility verification of an individual solution.

Common scaling of $(P,Y,b,\nu)$ by $c>0$ scales the Lyapunov level by $c$ and $P^{-1}$ by $c^{-1}$, leaving physical radii invariant. The normalization $\rho_T\leq1$ fixes this gauge and turns the componentwise radius objective into the LMI variable $\eta$.

\section{Structured nonlinear S-procedure}
\label{app:structured-derivation}

For the structured error model, let
\begin{equation}
\xi=[e^\top\ \Delta\phi^\top\ w^\top\ \bar v^\top\ \zeta_h^\top]^\top.
\end{equation}
Premultiplication and postmultiplication of the structured co-design LMI by $\xi^\top$ and $\xi$ yield
\begin{align}
0\geq{}&\dot V+aV-\mu_w\|w\|^2-\mu_v\|\bar v\|^2-\mu_h\|\zeta_h\|^2\nonumber\\
&+\tau\bigl(\gamma_\phi^2\|H_\phi e\|^2-\|\Delta\phi\|^2\bigr).
\end{align}
The final term is nonnegative under the structured incremental constraint; removing it gives the channel-wise dissipation inequality used in the main manuscript. This derivation also shows why the matrices $B_\phi$ and $H_\phi$ must remain in the LMI: replacing them by identity matrices recovers the more conservative full-norm sufficient condition that was not feasible for the tested benchmark.

\section{Sensitivity analyses}
\label{app:sensitivity}

This appendix examines how the nonlinear tube design changes under one-at-a-time variations of the sampling interval, state normalization, regional output-rate bound, and initial-error radius. Figure~\ref{fig:staleness-sensitivity} isolates the role of the sampling interval. Because $\bar\zeta_h=M_yh$, the held-output mismatch envelope grows linearly with $h$. Re-optimizing the observer partly compensates for this increase near the nominal value $h=0.01$~s, but the tube objective rises more rapidly for larger intervals, and the velocity-component half-widths remain the dominant absolute bounds. This confirms that output staleness is not merely a bookkeeping term but a material design limitation.

Figure~\ref{fig:assumption-sensitivity} addresses the remaining specified choices. Increasing the velocity reference scale from $k=4$ to $k=6$ substantially reduces the normalized objective, whereas the curve is essentially flat for $k\geq8$; the nominal choice $k=8$ therefore lies in the insensitive region. Varying $M_y$ around its nominal value produces only a modest change in the optimized objective, which supports the numerical robustness of the reported design without changing the regional status of the bound. By contrast, the objective grows approximately proportionally with the initial-error radius $r_0$, showing that the prescribed initial set is a primary determinant of the worst finite-horizon radius. These sweeps vary one quantity at a time and should not be interpreted as a joint robustness analysis.

\begin{figure*}[pos=htbp]
\centering
\includegraphics[width=0.94\textwidth]{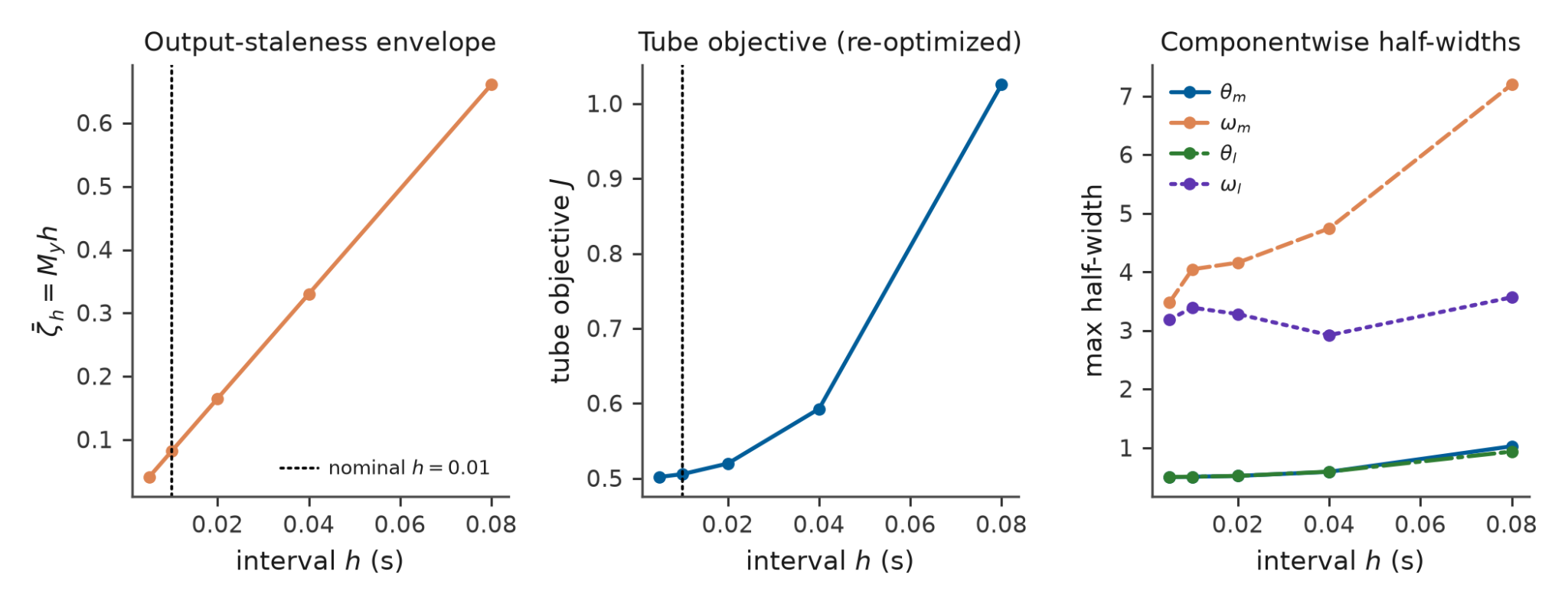}
\caption{Sensitivity of the nonlinear tube design to the sampling interval $h$, with the observer re-optimized for each value. The left panel shows the linear growth of the output-staleness envelope $\bar\zeta_h=M_yh$. The middle panel shows the corresponding optimized tube objective, and the right panel reports the componentwise maximum half-widths. The dotted vertical line marks the nominal interval $h=0.01$~s.}
\label{fig:staleness-sensitivity}
\end{figure*}

\begin{figure*}[pos=htbp]
\centering
\includegraphics[width=0.94\textwidth]{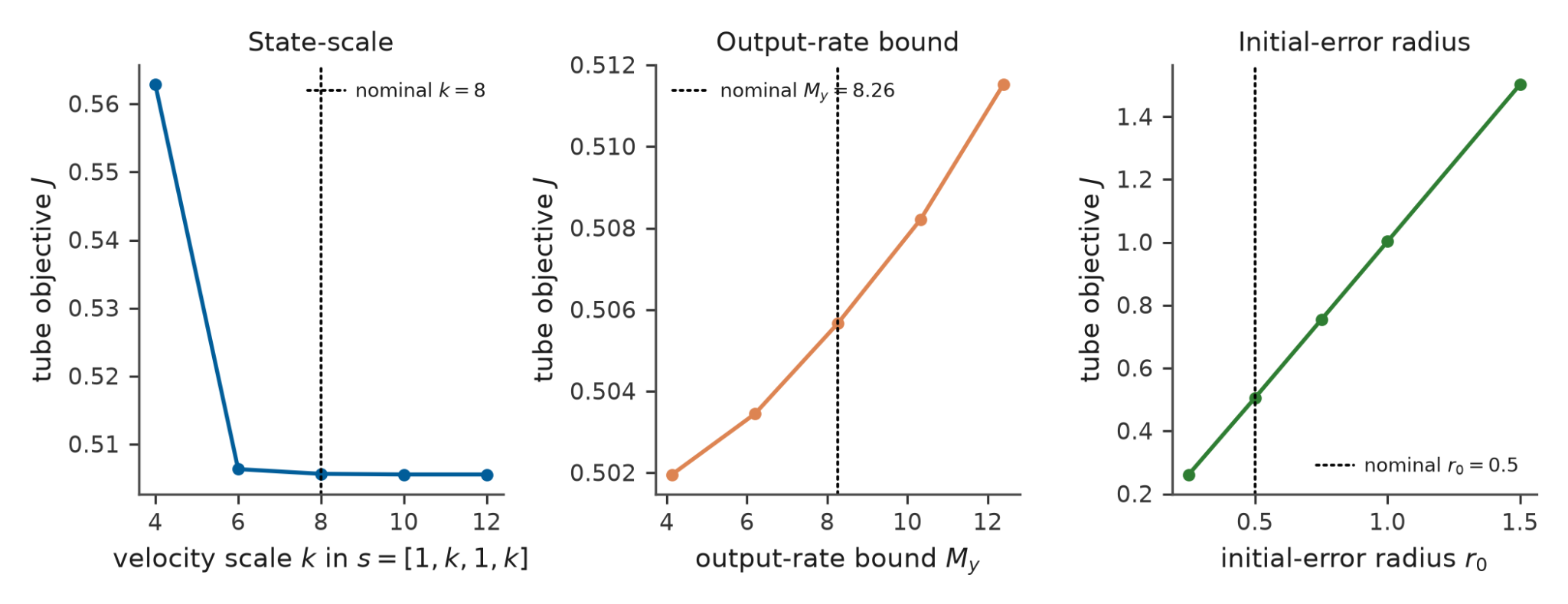}
\caption{Sensitivity of the nonlinear tube objective to selected modeling and normalization assumptions. The left panel varies the velocity reference scale $k$ in $s=(1,k,1,k)$ and shows that the objective is nearly constant once $k\geq8$. The middle panel varies the regional output-rate bound $M_y$ and shows weak sensitivity around the nominal value $8.26$. The right panel varies the initial-error radius $r_0$, for which the optimized objective grows approximately linearly. Dotted lines mark the nominal choices.}
\label{fig:assumption-sensitivity}
\end{figure*}

\section{Reproducibility protocol}
\label{app:reproducibility}
The public Zenodo replication package contains fixed configuration files, deterministic seed registries, solver and environment metadata, independently reconstructed residuals, and scripts that regenerate all reported figures and tables \cite{baranowski2026replication}. Full gains, Lyapunov matrices, and detailed LMI residual arrays are retained in machine-readable artifacts rather than typeset in the paper. The workflow preserves the run-index-to-seed mapping. The linear simulation uses $401$ state-grid points and $400$ measurement updates per trajectory; the corrected MHE implementation uses a growing horizon during startup and a time-aligned shifted arrival state once the window begins to slide. The archived record is available at DOI: \href{https://doi.org/10.5281/zenodo.21010764}{10.5281/zenodo.21010764}.

\end{document}